\documentclass[envcountsect,runningheads]{llncs}

\usepackage{amsmath}
\usepackage{amstext}
\usepackage{amssymb}
\usepackage{fp-frame}
\usepackage[latin1]{inputenc}
\usepackage{mathpartir}
\usepackage{fancyvrb}
\usepackage{moreverb}
\usepackage{stmaryrd}
\usepackage{enumerate}
\usepackage{thmtools,thm-restate}
\usepackage{comment}
\usepackage{booktabs,array}
\usepackage{rotating}
\usepackage{xcolor}
\usepackage{wrapfig}

\newcommand{\vb}[1]{\verb+#1+}

\newcommand{\defeq}{\triangleq}

\def\squareforqed{\hbox{\rlap{$\sqcap$}$\sqcup$}}
\def\qed{\ifmmode\squareforqed\else{\unskip\nobreak\hfil
\penalty50\hskip1em\null\nobreak\hfil\squareforqed
\parfillskip=0pt\finalhyphendemerits=0\endgraf}\fi}

\newcommand{\ttt}[1]{\texttt{#1}}
\newcommand{\gdesc}[1]{\text{\textit{#1}}}

\newcommand{\cod}[1]{\llbracket #1 \rrbracket}
\newcommand{\lcod}[2]{\llbracket #1 \rrbracket_{#2}}

\newcommand{\store}{m}

\newcommand{\extend}[3]{#1\{#2 \mapsto #3\}}

\newcommand{\assign}[2]{#1 := #2}

\newcounter{topiccounter}
\newcommand{\redx}{\rightarrow}
\newcommand{\redxs}{\redx^*}

\newcommand{\elet}[3]{\mathrm{let}\ #1 = #2\ \mathrm{in}\ #3}

\newcommand{\dom}{\mathrm{dom}}

\newcommand{\latel}{\varsigma}
\newcommand{\hilab}{\mathrm{High}}
\newcommand{\lolab}{\mathrm{Low}}

\newcommand{\minifed}{\mathit{Overture}}
\newcommand{\minicat}{\minifed}
\newcommand{\fedprot}{\minifed}
\newcommand{\metaprot}{\mathit{Prelude}}

\newcommand{\prog}{\pi}

\newcommand{\bop}{\ \mathit{binop}\ }

\newcommand{\mems}{\mathit{mems}}

\newcommand{\margd}[2]{{#1}_{#2}}
\newcommand{\condd}[3]{#1_{({#2}|{#3})}}
\newcommand{\progd}{\mathrm{BD}}
\newcommand{\progtt}{\mathrm{BD}}
\newcommand{\vars}{\mathit{vars}}
\newcommand{\iov}{\mathit{iovars}}
\newcommand{\flips}{\mathit{flips}}

\newcommand{\fedcat}{\minifed}

\newcommand{\idealf}{\mathcal{F}}

\def\TirName#1{\text{\sc #1}}

\newcommand{\flip}[2]{\ttt{flip[}#1\ttt{,}#2\ttt{]}}
\newcommand{\secret}[2]{\ttt{s[}#1\ttt{,}#2\ttt{]}}

\renewcommand{\elet}[3]{\ttt{let}\ #1\ \ttt{=}\ #2\ \ttt{in}\ #3}

\renewcommand{\redx}{\xrightarrow{}}
\renewcommand{\redxs}{\xrightarrow{}^{*}}

\newcommand{\secrets}{\mathit{secrets}}

\newcommand{\views}{\mathit{views}}

\newcommand{\cid}{\iota}

\newcommand{\OT}[3]{\ttt{OT(} #1 \ttt{,}\ #2 \ttt{,}\ #3 \ttt{)}}

\newcommand{\mux}[3]{\ttt{mux(} #1 \ttt{,}\ #2 \ttt{,}\ #3 \ttt{)}}
\newcommand{\codebase}{\mathcal{C}}

\newcommand{\flab}{\ell}
\newcommand{\be}{\varepsilon}
\newcommand{\instr}{\mathbf{c}}

\newcommand{\runs}{\mathit{runs}}

\newcommand{\concat}{\ttt{++}}

\newcommand{\sep}[3]{#1 \vdash #2 * #3}
\newcommand{\notsep}[3]{#1 \not\vdash #2 * #3}
\newcommand{\condp}[3]{#1|#2 \vdash #3}
\newcommand{\notcondp}[3]{#1|#2 \not\vdash #3}

\newcommand{\condsep}[4]{\condp{#1}{#2}{#3 * #4}}
\newcommand{\notcondsep}[4]{\notcondp{#1}{#2}{#3 * #4}}

\newcommand{\pmf}{\mathcal{P}}

\renewcommand{\flip}[1]{\ttt{r[}#1\ttt{]}}
\newcommand{\locflip}{\ttt{r[}\mathtt{local}\ttt{]}}
\renewcommand{\secret}[1]{\ttt{s[}#1\ttt{]}}

\newcommand{\mesg}[1]{\ttt{m[}#1\ttt{]}}
\newcommand{\outkw}{\ttt{out}}
\newcommand{\out}[1]{\elab{\outkw}{#1}}
\newcommand{\rvl}[1]{\ttt{p[}#1\ttt{]}}

\newcommand{\elab}[2]{#1\ttt{@}#2}
\newcommand{\eassign}[4]{\elab{#1}{#2} := \elab{#3}{#4}}
\newcommand{\xassign}[3]{#1 := \elab{#2}{#3}}
\newcommand{\pubout}[3]{\out{#1} := \elab{#2}{#3}}
\newcommand{\reveal}[3]{\rvl{#1} := \elab{#2}{#3}}

\newcommand{\adversary}{\mathcal{A}}
\newcommand{\aredx}{\redx_{\adversary}}
\newcommand{\aredxs}{\redxs_{\adversary}}
\newcommand{\arewrite}{\mathit{rewrite}_{\adversary}}
\newcommand{\cinputs}{V_{C \rhd H}}
\newcommand{\houtputs}{V_{H \rhd C}}
\newcommand{\aruns}{\mathit{runs}_\adversary}
\newcommand{\botruns}{\mathit{runs}_{\adversary,\bot}}

\renewcommand{\store}{\sigma}

\renewcommand{\runs}{\mathit{runs}}

\newcommand{\fcod}[1]{\lcod{#1}{}}
\renewcommand{\flips}{\mathit{rands}}

\newcommand{\ftimes}{*}
\newcommand{\fplus}{+}
\newcommand{\fminus}{-}
\newcommand{\macgv}[1]{\langle #1 \rangle}

\newcommand{\preproc}{\mathit{preproc}}
\newcommand{\assert}[1]{\ttt{assert(}#1\ttt{)}}
\newcommand{\mv}{\nu}

\newcommand{\sx}[2]{\elab{\secret{#1}}{#2}}
\newcommand{\mx}[2]{\elab{\mesg{#1}}{#2}} 
 
\newcommand{\rx}[2]{\elab{\flip{#1}}{#2}}

\newcommand{\eqcast}[2]{#1\ \ttt{as}\ #2}

\newcommand{\itj}[3]{\vdash_{#1} #2 : #3}
\newcommand{\ipj}[3]{#1 \vdash #2 : #3}

\newcommand{\cty}[2]{c(#1,#2)}
\newcommand{\setit}[1]{\{ #1 \}}
\newcommand{\ty}{T}
\newcommand{\ity}[2]{#1 \cdot #2}

\newcommand{\eqs}{\mathit{E}}
\newcommand{\toeq}[1]{\lfloor #1 \rfloor}
\newcommand{\eop}{\equiv}

\newcommand{\seclev}{\mathcal{L}}

\newcommand{\tsig}{\mathrm{sig}}

\newcommand{\hty}[5]{\{ #1 \}\ #2,#3 \cdot #4\  \{ #5 \} }
\newcommand{\dht}[6]{\Pi #1 . \hty{#2}{#3}{#4}{#5}{#6}}
\newcommand{\mtj}[6]{\vdash #1 : \hty{#2}{#3}{#4}{#5}{#6}}
\newcommand{\atj}[3]{\Vdash #1 : (#2,#3)}
\newcommand{\eqj}[4]{#1,#2 \vdash #3 : #4}
\newcommand{\cpj}[4]{#1,#2 \vdash #3 : #4}
\newcommand{\leakclose}[3]{#1 \vdash #2 \leadsto #3}
\newcommand{\leakj}[3]{#1,#2 \vdash_{\mathit{leak}} #3}
\newcommand{\cheatj}[3]{#1 \underset{#2}{\leadsto} #3}

\renewcommand{\redx}{\Rightarrow}
\renewcommand{\redxs}{\redx}
\newcommand{\abort}{\bot}

\newcommand{\cmd}{\instr}
\renewcommand{\OT}[4]{\ttt{OT(}\elab{#1}{#2},#3,#4 \ttt{)}}
\newcommand{\eqspre}{\eqs_{\mathit{pre}}}
\newcommand{\macbdoz}[1]{\psi_{\mathit{BDOZ}}(#1)}

\newcommand{\notg}[1]{\breve{#1}}
\newcommand{\fresh}{\mathit{fresh}}
\newcommand{\fieldp}[1]{\mathbb{F}_{#1}}
\newcommand{\precond}{\mathit{precond}}
\newcommand{\postcond}{\mathit{postcond}}

\newcommand{\compwrapfig}
{
  \begin{wrapfigure}{r}{0cm}
    \begin{tabular}{lccccccc}
      & 
      \begin{sideways} probabilistic language \end{sideways} &
      \begin{sideways} probabilistic conditioning \end{sideways} & 
      \begin{sideways} low-level protocols \end{sideways} & 
      \begin{sideways} passive security \end{sideways} & 
      \begin{sideways} malicious security \end{sideways}& 
      \begin{sideways} hyperproperties \end{sideways}& 
      \begin{sideways} automation \end{sideways}\\\hline\hline
      Haskell EDSL \cite{6266151} & \checkmark &  & \checkmark  & \checkmark & & \checkmark & \checkmark \\\hline
      MPC in SecreC \cite{almeida2018enforcing} & \checkmark & \checkmark &   & \checkmark & & \checkmark & \checkmark \\\hline
      $\lambda_{\text{obliv}}$ \cite{darais2019language} & \checkmark & & \checkmark & & & \checkmark & \checkmark \\\hline
      PSL \cite{barthe2019probabilistic} & \checkmark & & \checkmark & & & & \\\hline
      Lilac \cite{li2023lilac} & \checkmark & \checkmark & & & & & \\\hline
      Wys$^*$ \cite{wysstar} & & & & \checkmark & & \checkmark & \checkmark\\\hline
      Viaduct \cite{10.1145/3453483.3454074,viaduct-UC} & & & & \checkmark & \checkmark & \checkmark & \checkmark\\\hline
      MPC in EasyCrypt \cite{8429300} &  \checkmark &  \checkmark &  \checkmark & \checkmark & \checkmark & \checkmark & \\\hline
      $\metaprot$/$\minifed$ \cite{skalka-near-ppdp24} & \checkmark & \checkmark & \checkmark & \checkmark & \checkmark & \checkmark & * \\\hline
      This work & \checkmark & \checkmark & \checkmark & \checkmark & \checkmark & \checkmark & \checkmark\\
      \hline
    \end{tabular}
    \caption{Comparison of systems for verification of MPC security in PLs. * indicates limited support for automation.}
    \label{fig-comp-wrap}
  \end{wrapfigure}
}

\newcommand{\minicatsyntaxfig}{
\begin{fpfig}[t]{Syntax of $\minicat$}{fig-minicat-syntax}
\small{
$$
    \begin{array}{rcl@{\hspace{4mm}}r}
      \multicolumn{4}{l}{v \in \mathbb{F}_p,\ w \in \mathrm{String},\ \cid \in \mathrm{Clients} \subset  \mathbb{N} }\\[2mm] %, \bop \in \{ \eand, \eor, \exor \}} \\[2mm]
      \be &::=& v \mid \flip{w} \mid \secret{w} \mid \mesg{w} \mid \rvl{w} \mid \be \fminus \be \mid \be \fplus \be \mid \be \ftimes \be \mid \OT{\be}{\cid}{\be}{\be} & \textit{expressions}\\[1mm]
      x &::=& \elab{\flip{w}}{\cid} \mid \elab{\secret{w}}{\cid} \mid \elab{\mesg{w}}{\cid} \mid  \rvl{w} \mid \out{\cid} & \textit{variables} \\[1mm]
      \prog &::=& \eassign{\mesg{w}}{\cid}{\be}{\cid} \mid \reveal{w}{\be}{\cid} \mid \pubout{\cid}{\be}{\cid} \mid \prog;\prog & \textit{protocols}
    \end{array}
$$}
\end{fpfig}    
}

\newcommand{\minicatredxfig}{
\begin{fpfig}[t]{Semantics of $\minicat$ expressions (T) and programs (B).}{fig-minicat-redx}
\small{
 $$
  \begin{array}{c@{\hspace{5mm}}c}
  \begin{array}{rcl}
    \lcod{\store, v}{\cid} &=& v\\
    \lcod{\store, \be_1 \fplus \be_2}{\cid} &=& \fcod{\lcod{\store, \be_1}{\cid} \fplus \lcod{\store, \be_2}{\cid}}\\ 
    \lcod{\store, \be_1 \fminus \be_2}{\cid} &=& \fcod{\lcod{\store, \be_1}{\cid} \fminus \lcod{\store, \be_2}{\cid}}\\ 
    \lcod{\store, \be_1 \ftimes \be_2}{\cid} &=& \fcod{\lcod{\store, \be_1}{\cid} \ftimes \lcod{\store, \be_2}{\cid}}\\
  \end{array} & 
  \begin{array}{rcl}
    \lcod{\store, \flip{w}}{\cid} &=& \store(\elab{\flip{w}}{\cid})\\
    \lcod{\store, \secret{w}}{\cid} &=& \store(\elab{\secret{w}}{\cid})\\
    \lcod{\store, \mesg{w}}{\cid} &=& \store(\elab{\mesg{w}}{\cid})\\
    \lcod{\store, \rvl{w}}{\cid} &=& \store(\rvl{w})\\
  \end{array}
  \end{array}
  $$

\begin{mathpar}
  (\store, \xassign{x}{\be}{\cid}) \redx \extend{\store}{x}{\lcod{\store,\be}{\cid}}
  
  \inferrule
      {(\store_1,\prog_1) \redx \store_2 \\ (\store_2,\prog_2) \redx \store_3 }
      {(\store_1,\prog_1;\prog_2) \redx \store_3}
\end{mathpar}
}
\end{fpfig}
}

\newcommand{\minicataredxfig}{
\begin{fpfig}[t]{Adversarial semantics of $\minicat$.}{fig-minicat-aredx}
\small{
\begin{mathpar}
  \inferrule
      { \cid \in H }
      { (\store, \xassign{x}{\be}{\cid}) \aredx \extend{\store}{x}{\lcod{\store,\be}{\cid}} }
      
  \inferrule
      {\cid \in C }
      { (\store, \xassign{x}{\be}{\cid}) \aredx \extend{\store}{x}{\lcod{\arewrite(\store_C,\be)}{\cid}}}
      
  \inferrule
      {\lcod{\store,\be_1}{\cid} = \lcod{\store,\be_2}{\cid}  \text{\ or\ } \cid \in C}
      { (\store,\elab{\assert{\be_1 = \be_2}}{\cid}) \aredx \store }
      
  \inferrule
      {\lcod{\store,\be_1}{\cid} \ne \lcod{\store,\be_2}{\cid}}
      {(\store,\elab{\assert{\be_1 = \be_2}}{\cid}) \aredx \abort}
  
  \inferrule
      {(\store_1,\prog_1) \aredx \store_2 \\ (\store_2,\prog_2) \aredx \store_3 }
      {(\store_1,\prog_1;\prog_2) \aredx \store_3}

  \inferrule
      {(\store_1,\prog_1) \aredx \abort}
      {(\store_1,\prog_1;\prog_2) \aredx \abort}
      
  \inferrule
      {(\store_1,\prog_1) \aredx \store_2 \\ (\store_2,\prog_2) \aredx \abort }
      {(\store_1,\prog_1;\prog_2) \aredx \store_2}
\end{mathpar}}
\end{fpfig}
}

\newcommand{\cpjfig}{
\begin{fpfig}[t]{Syntax and Derivation Rules for $\minicat$ Confidentiality Types}{fig-cpj}
\small{
$$
\begin{array}{rcl@{\hspace{3mm}}l}
  t &::=& x \mid \cty{x}{T} \\
  \ty & \in & 2^{t} & \gdesc{confidentiality types}\\
  \Gamma &::=& \varnothing \mid \Gamma; x : \ty & \gdesc{confidentiality type environments}
\end{array} 
$$
\medskip
\begin{mathpar}
  \inferrule[DepTy]
  {}
  {\eqj{\varnothing}{\eqs}{\phi}{\vars(\phi)}}
  
  \inferrule[Encode]
  {\eqs \models \phi \eop \phi' \oplus \rx{w}{\cid} \\
   \oplus \in \{ \fplus,\fminus \}\\
   \eqj{R}{\eqs}{\phi'}{\ty}}
  {\eqj{R;\{ \rx{w}{\cid} \}}{\eqs}{\phi}{\setit{\cty{\rx{w}{\cid}}{\ty}}}}
\end{mathpar}

\begin{mathpar}
  \inferrule[Send]
            {\eqj{R}{\eqs}{\phi}{\ty}}
            {\cpj{R}{\eqs}{x \eop \phi}{(x : \ty)}}
            
  \inferrule[Seq]
            {\cpj{R_1}{\eqs}{\phi_1}{\Gamma_1}\\
             \cpj{R_2}{\eqs}{\phi_2}{\Gamma_2}}
            {\cpj{R_1;R_2}{\eqs}{\phi_1 \wedge \phi_2}{\Gamma_1;\Gamma_2}}
\end{mathpar}
}
\end{fpfig}
}

\newcommand{\leakjfig}{
\begin{fpfig}[t]{Dependencies in Views: Derivation Rules}{fig-leakj}
\small{
\begin{mathpar}
  \inferrule
      {\leakclose{\Gamma}{\ty_1}{\ty_2} \\ \leakclose{\Gamma}{\ty_2}{\ty_3}}
      {\leakclose{\Gamma}{\ty_1}{\ty_3}}

      \leakclose{\Gamma}{\ty \cup \setit{\mx{w}{\cid}}}{\ty \cup \Gamma(\mx{w}{\cid})}

      \leakclose{\Gamma}{\ty_1 \cup \setit{x, \cty{x}{\ty_2}}}{\ty_1\cup\ty_2}
\end{mathpar}

\begin{mathpar}
  \inferrule
      {}
      {\leakj{\Gamma}{\varnothing}{\varnothing}}

\inferrule
      {\leakj{\Gamma}{M}{\ty'} \\ \leakclose{\Gamma}{\ty'\cup\setit{x}}{\ty}}
      {\leakj{\Gamma}{M \cup \setit{x}}{\ty}}
\end{mathpar}
}
\end{fpfig}
}

\newcommand{\ipjfig}{
\begin{fpfig}[t]{Syntax and derivation rule of $\minicat$ integrity types}{fig-ipj}
\small{
$$
\begin{array}{rcl@{\hspace{4mm}}l}
  \latel &::=& \hilab \mid \lolab & \gdesc{integrity labels} \\
  \Delta &::=& \varnothing \mid \Delta; x : \ity{\cid}{V} & \gdesc{integrity type environments}
\end{array} 
$$

\begin{mathpar}
  \inferrule[Value]
  {}
  {\itj{\cid}{v}{\varnothing}}
  
  \inferrule[Secret]
  {}
  {\itj{\cid}{\secret{w}}{\varnothing}}
  
  \inferrule[Rando]
  {}
  {\itj{\cid}{\flip{w}}{\varnothing}}
  
  \inferrule[Mesg]
  {}
  {\itj{\cid}{\mesg{w}}{\setit{\mx{w}{\cid}}}}
    
  \inferrule[PubM]
  {}
  {\itj{\cid}{\rvl{w}}{\setit{\rvl{w}}}}

  \inferrule[Binop]
  {\itj{\cid}{\be_1}{V_1} \\
   \itj{\cid}{\be_2}{V_2} \\ \oplus \in \{ \fplus,\fminus,\ftimes \}}
  {\itj{\cid}{\be_1 \oplus \be_2}{V_1 \cup V_2}}
\end{mathpar}

\begin{mathpar}
  \inferrule[Send]
            {\itj{\cid}{\be}{V}}
            {\ipj{\eqs}{\xassign{x}{\be}{\cid}}{(x : \ity{\cid}{V})}}
             
  \inferrule[Seq]
            {\ipj{\eqs}{\prog_1}{\Delta_1}\\
             \ipj{\eqs}{\prog_2}{\Delta_2}}
            {\ipj{\eqs}{\prog_1;\prog_2}{\Delta_1;\Delta_2}}

  \inferrule[MAC]
            {\eqs \models \toeq{\elab{\assert{\macbdoz{w}}}{\cid}}}
            {\ipj{\eqs}{\elab{\assert{\macbdoz{w}}}{\cid}}{(\mx{w\ttt{s}}{\cid}: \ity{\cid}{\varnothing})}}
\end{mathpar}
}
\end{fpfig}
}

\newcommand{\cheatjfig}{
\begin{fpfig}[t]{Assigning integrity labels to variables}{fig-cheatj}
\small{
\begin{mathpar}
  \inferrule
      {}
      {\cheatj{\varnothing}{H,C}{\seclev_{H,C}}}
      
  \inferrule
      {\cheatj{\Delta}{H,C}{\seclev} \\ \cid \in H}
      {\cheatj{\Delta; x : \ity{\cid}{V}}{H,C}{\extend{\seclev}{x}{\hilab \wedge (\bigwedge_{x \in V} \seclev_2(x))}}}
      
  \inferrule
      {\cheatj{\Delta}{H,C}{\seclev} \\ \cid \in C}
      {\cheatj{\Delta; x : \ity{\cid}{V}}{H,C}{\extend{\seclev}{x}{\lolab}}}
\end{mathpar}
}
\end{fpfig}
}

\newcommand{\metaprotsyntaxfig}{
  \begin{fpfig}[t]{$\metaprot$ syntax.}{fig-metaprotsyntax}
\small{
$$
\begin{array}{rcl@{\hspace{4mm}}r}
  \multicolumn{3}{l}{\flab \in \mathrm{Field},\   y \in \mathrm{EVar}, \  f \in \mathrm{FName}}\\[1mm]
  e &::=& \mv \mid \flip{e} \mid \secret{e} \mid \mesg{e} \mid \rvl{e} \mid \outkw \mid e \bop e 
   \mid y \mid & \gdesc{expressions}\\
  & & e.\flab \mid \elab{e}{e} \mid \elet{y}{e}{e} \mid  f(e,\ldots,e) \mid \{ \flab = e; \ldots; \flab = e \} \\[1mm]
  \cmd &::=& %\msend{e}{e}{e}{e} \mid \reveal{e}{e}{e} \mid \pubout{e}{e}{e} \mid
  \assign{e}{e} \mid f(e,\ldots,e) \mid
  \elet{y}{e}{\instr} \mid  \cmd;\cmd & \gdesc{instructions}\\
         & & \elab{\assert{e = e}}{e} \mid \eqcast{\mx{e}{e}}{\notg{\phi}} \\[1mm] %\pre{\eqs} \mid \post{\eqs} \\[1mm]
  \bop &::=& \fplus \mid \fminus \mid \ftimes \mid \concat  \\[1mm]% \textit{operators}\\[2mm]
  \mv &::=& w \mid \cid \mid \be \mid x \mid \{ \flab = \mv;\ldots;\flab = \mv \} 
  & \gdesc{values}\\[1mm] % \mid \ttt{()} \\[1mm] %& \textit{values}\\[2mm]
  \mathit{fn} &::=& f(y,\ldots,y) \{ e \} \mid  f(y,\ldots,y) \{ \cmd \} & \textit{functions}
\end{array}
$$
}
\end{fpfig}
}

\newcommand{\metaprotexprsemanticsfig}{
  \begin{fpfig}[t]{Semantics of $\metaprot$ expressions.}{fig-metaprotexprsemantics}
\small{
  \begin{mathpar}
  \inferrule
      {e_1 \redx \mv \\ e_2[\mv/y] \redx \mv'}
      {\elet{y}{e_1}{e_2} \redx \mv'}
      
  \inferrule
      {e_1 \redx \be \\ e_2 \redx \cid}
      {\elab{e_1}{e_2} \redx \elab{\be}{\cid}}

  \inferrule
      {\codebase(f) = y_1,\ldots,y_n,\ e \\ e_1 \redx \mv_1 \cdots e_n \redx \mv_n \\
        e[\mv_1/y_1]\cdots[\mv_n/y_n] \redx \mv}
      {f(e_1,\ldots,e_n) \redx \mv}

  \inferrule
      {e_1 \redx \mv_1 \\ \cdots \\ e_n \redx \mv_n }
      {\{ \flab_1 = e_1; \ldots; \flab_n = e_n \} \redx \{ \flab_1 = \mv_1; \ldots; \flab_n = \mv_n \} }

  \inferrule
      {e \redx \{\ldots; \flab = \mv; \ldots\}}
      {e.\flab \redx \mv}

  \inferrule
      {e_1 \redx w_1 \\ e_2 \redx w_2}
      {e_1 \concat e_2 \redx w_1w_2}

  \inferrule
      {e \redx w}
      {\mesg{e} \redx \mesg{w}}
      
  \inferrule
      {e_1 \redx \be_1 \\ e_2 \redx \be_2}     
      {e_1 \fplus e_2 \redx \be_1 \fplus \be_2}
\end{mathpar}
}
\end{fpfig}
}

\newcommand{\metaprotinstrsemanticsfig}{
  \begin{fpfig}[t]{Semantics of $\metaprot$ instructions.}{fig-metaprotinstrsemantics}
\begin{mathpar}
  \inferrule
      {e_1 \redx x \\ e_2 \redx \elab{\be}{\cid}}
      {\assign{e_1}{e_2} \redx \xassign{x}{\be}{\cid}}

  \inferrule
      {e_1 \redx \be_1 \\ e_2 \redx \be_2 \\ e_3 \redx \cid}
      {\elab{\assert{e_1 = e_2}}{e_3} \redx \elab{\assert{\be_1 = \be_2}}{\cid}}

  \inferrule
      {\codebase(f) = y_1,\ldots,y_n, \instr \\e_1 \redx \mv_1 \ \cdots \ e_n \redx \mv_n \\
        \instr[\mv_1/y_1]\cdots[\mv_n/y_n] \redx \prog
      }
      {f(e_1,\ldots,e_n) \redx \prog}
      
  \inferrule
      {e \redx \mv \\ \instr[\mv/y] \redx \prog}
      {\elet{y}{e}{\instr} \redx \prog}

  \inferrule
      {e_1 \redx \prog_1 \\ e_2 \redx \prog_2}
      {e_1;e_2 \redx \prog_1;\prog_2}
\end{mathpar}
\end{fpfig}
}

\newcommand{\atjfig}{
  \begin{fpfig}[t]{Algorithmic type judgements for $\minicat$.}{fig-atj}
\small{
\begin{mathpar}
  \atj{x}{\varnothing}{\setit{x}}

  \inferrule
  {\atj{\phi}{R}{\ty} \\ \rx{w}{\cid}\not\in R \\ \oplus \in \setit{\fplus,\fminus}}
  {\atj{\phi \oplus \rx{w}{\cid}}{R \cup \setit{\rx{w}{\cid}}}{\setit{\cty{\rx{w}{\cid}}{\ty}}}}

  \inferrule
  {\atj{\phi_1}{R_1}{\ty_1} \\
   \atj{\phi_2}{R_2}{\ty_2} \\ \oplus \in \{ \fplus,\fminus,\ftimes \}}
  {\atj{\phi_1 \oplus \phi_2}{R_1;R_2}{\ty_1 \cup \ty_2}}
\end{mathpar}
}
\end{fpfig}
}

\newcommand{\notgfig}{
\begin{fpfig}[t]{Evaluation of expressions within types and constraints.}{fig-notg}
$$
\begin{array}{cc}
  \begin{array}{rcl}
    \notg{x} &::=& \elab{\flip{e}}{e} \mid \elab{\secret{e}}{e} \mid \elab{\mesg{e}}{e} \mid \rvl{e} \mid \out{e}\\
  \notg{\phi} &::=& \notg{x} \mid \notg{\phi} \fplus \notg{\phi} \mid \notg{\phi} \fminus \notg{\phi} \mid \notg{\phi} \ftimes \notg{\phi} \\
  \notg{\eqs} &::=& \notg{\phi} \eop \notg{\phi} \mid \notg{\eqs} \wedge \notg{\eqs} \\
  \notg{X} &\in& 2^{\notg{x}}
\end{array}& \qquad
\begin{array}{rcl}
  \notg{t} &::=& e \mid \cty{e}{\notg{\ty}} \\
  \notg{\ty} & \in & 2^{\notg{t}}\\
  \notg{\Gamma} &::=& \varnothing \mid \notg\Gamma; e : \notg{\ty}\\
  \notg\Delta &::=& \varnothing \mid \notg\Delta; e : \ity{e}{\notg{V}}
\end{array}
\end{array}
$$

\begin{mathpar}
  \inferrule
      {\notg{\phi_1} \redx \phi_1 \\ \notg{\phi_2} \redx \phi_2}     
      {\notg{\phi_1} \ftimes \notg{\phi_2} \redx \phi_1 \ftimes \phi_2}

  \inferrule
      {\notg{\phi_1} \redx \phi_1 \\ \notg{\phi_2} \redx \phi_2}
      {\notg{\phi_1} \eop \notg{\phi_2} \redx \phi_1 \eop \phi_2}

  \inferrule
      {\notg{\eqs_1} \redx \eqs_1 \\ \notg{\eqs_2} \redx \eqs_2 }
      {\notg{\eqs_1} \wedge \notg{\eqs_2} \redx \eqs_1 \wedge \eqs_2}
\end{mathpar}

\begin{mathpar}
  \inferrule
      {e \redx x \\ \notg{\ty} \redx \ty}
      {\cty{e}{\notg{\ty}} \redx \cty{x}{\ty}}
      
  \inferrule
      {\notg{t_1} \redx t_1 \\ \cdots \\ \notg{t_n} \redx t_n}
      {\setit{\notg{t_1},\ldots,\notg{t_n}} \redx \setit{ t_1,\ldots,t_n }}

  \inferrule
      {\notg{\Gamma} \redx \Gamma \\ e \redx x \\ \notg{\ty} \redx \ty }
      {\notg{\Gamma}; e : \notg{\ty} \redx \Gamma; x : \ty }

  \inferrule
      {\notg{\Delta} \redx \Delta \\ e_1 \redx x  \\ e_2 \redx \cid \\ \notg{V} \redx V}
      {\notg{\Delta}; e_1 : \ity{e_2}{\notg{V}} \redx \Delta; x : \ity{\cid}{V} }

  \inferrule
      {\notg{\eqs_1} \redx \eqs_1 \\ \notg{\Gamma} \redx \\ \notg{R} \redx R
        \\ \notg{\Delta} \redx \Delta \\ \notg{\eqs_2} \redx \eqs_2}
      {\hty{\notg{\eqs_1}}{\notg{\Gamma}}{\notg{R}}{\notg{\Delta}}{\notg{\eqs_2}} \redx
        \hty{\eqs_1}{\Gamma}{R}{\Delta}{\eqs_2}}
\end{mathpar}
    
\end{fpfig}
}

\newcommand{\mtjfig}{
\begin{fpfig}[h]{$\metaprot$ type derivation rules for instructions.}{fig-mtj}
\begin{mathpar}
  \inferrule[Mesg]
            {\assign{e_1}{e_2} \redx \xassign{x}{\be}{\cid}  \\ \atj{\toeq{\elab{\be}{\cid}}}{R}{\ty} \\
              \itj{\cid}{\be}{V}}
            {\mtj{\assign{e_1}{e_2}}{\eqs}{(x:\ty)}{R}{(x : \ity{\cid}{V})}{\eqs \wedge x \eop \toeq{\elab{\be}{\cid}}}}

  \inferrule[Encode]
            {\mx{e_1}{e_2} \redx x \\ \notg{\phi} \redx \phi \\
              \eqs \models x \eop \phi\\
              \atj{\phi}{R}{\ty}}
            {\mtj{\eqcast{\mx{e_1}{e_2}}{\notg{\phi}}}{\eqs}{(x : \ty)}{R}{\varnothing}{\eqs}}

  \inferrule[Assert]
            {\elab{\assert{e_1 = e_2}}{e_3} \redx \elab{\assert{\be_1 = \be_2}}{\cid} \\
             \ipj{\eqs}{\elab{\assert{\be_1 = \be_2}}{\cid}}{\Delta}}
            {\mtj{\elab{\assert{e_1 = e_2}}{e_3}}{\eqs}{\varnothing}{\varnothing}{\Delta}{\eqs}}
            
  \inferrule[App]
            {f : \dht{y_1,\ldots,y_n}{\notg{\eqs_1}}{\notg{\Gamma}}{\notg{R}}{\notg{\Delta}}{\notg{\eqs_2}} \\
              e_1 \redx \mv_1\ \cdots\ e_n \redx \mv_n \\
              (\hty{\notg{\eqs_1}}{\notg{\Gamma}}{\notg{R}}{\notg{\Delta}}{\notg{\eqs_2}})[\mv_1/y_1]\cdots[\mv_n/y_n] \redx
                    \hty{\eqs_1}{\Gamma}{R}{\Delta}{\eqs_2} \\
              \eqs \models \eqs_1}
            {\mtj{f(e_1,\ldots,e_n)}{\eqs}{\Gamma}{R}{\Delta}{\eqs \wedge \eqs_2}}

  \inferrule[Seq]          
            {\mtj{\prog_1}{\eqs_1}{\Gamma_1}{R_1}{\Delta_1}{\eqs_2} \\
             \mtj{\prog_2}{\eqs_2}{\Gamma_2}{R_2}{\Delta_2}{\eqs_3}}
            {\mtj{\prog_1;\prog_2}{\eqs_1}{\Gamma_1;\Gamma_2}{R_1;R_2}{\Delta_1;\Delta_2}{\eqs_3}}
\end{mathpar}
\end{fpfig}
}

\newcommand{\mtjfnfig}{
\begin{fpfig}[h]{$\metaprot$ type derivation rules for function definitions.}{fig-mtjfn}
\begin{mathpar}
   \inferrule[Fn]
            {\codebase(f) = y_1,\ldots,y_n, \instr \\
              \mtj{\instr[\mv_1/y_1]\cdots[\mv_n/y_n]}{\eqs_1}{\Gamma}{R}{\Delta}{\eqs_2}\\
              \fresh(\mv_1,\ldots,\mv_n) \\
              (\hty{\notg{\eqs_1}}{\notg{\Gamma}}{\notg{R}}{\notg{\Delta}}{\notg{\eqs_2}})[\mv_1/y_1]\cdots[\mv_n/y_n]  \redx
                    \hty{\eqs_1}{\Gamma}{R}{\Delta}{\eqs_2} }
            {f : \dht{y_1,\ldots,y_n}{\notg{\eqs_1}}{\notg{\Gamma}}{\notg{R}}{\notg{\Delta}}{\notg{\eqs_2}}}

  \inferrule[FnPre]
            {f : \dht{y_1,\ldots,y_n}{\notg{\eqs}}{\notg{\Gamma}}{\notg{R}}{\notg{\Delta}}{\notg{\eqs_2}} \\
              \precond(f) = \notg{\eqs_1} \\
              \fresh(\mv_1,\ldots,\mv_n) \\
              \notg{\eqs}[\mv_1/y_1]\cdots[\mv_n/y_n]  \redx \eqs \\
              \notg{\eqs_1}[\mv_1/y_1]\cdots[\mv_n/y_n]  \redx \eqs_1 \\
              \eqs_1 \models \eqs             
            }
            {f : \dht{y_1,\ldots,y_n}{\notg{\eqs_1}}{\notg{\Gamma}}{\notg{R}}{\notg{\Delta}}{\notg{\eqs_2}}}

  \inferrule[FnPost]
            {f : \dht{y_1,\ldots,y_n}{\notg{\eqs_1}}{\notg{\Gamma}}{\notg{R}}{\notg{\Delta}}{\notg{\eqs}} \\
              \postcond(f) = \notg{\eqs_2} \\
              \fresh(\mv_1,\ldots,\mv_n) \\
              \notg{\eqs}[\mv_1/y_1]\cdots[\mv_n/y_n]  \redx \eqs \\
              \notg{\eqs_2}[\mv_1/y_1]\cdots[\mv_n/y_n]  \redx \eqs_2 \\
              \eqs \models \eqs_2              
            }
            {f : \dht{y_1,\ldots,y_n}{\notg{\eqs_1}}{\notg{\Gamma}}{\notg{R}}{\notg{\Delta}}{\notg{\eqs_2}}}
\end{mathpar}
\end{fpfig}
}

\titlerunning{Security Types for Low-Level MPC}

\title{SMT-Boosted Security Types for Low-Level MPC}

\author{
Christian Skalka
\and
Joseph P.~Near}

\institute{
University of Vermont, Burlington VT 05405, USA 
\email{ceskalka@uvm.edu,jnear@uvm.edu}
}

\newcommand{\theabstract}{
  \begin{abstract}
  Secure Multi-Party Computation (MPC) is an important enabling
  technology for data privacy in modern distributed applications. We
  develop a new type theory to automatically enforce correctness,
  confidentiality, and integrity properties of protocols written in
  the \emph{Prelude/Overture} language framework. Judgements in the
  type theory are predicated on SMT verifications in a theory of
  finite fields, which supports precise and efficient analysis. Our
  approach is automated, compositional, scalable, and generalizes to arbitrary
  prime fields for data and key sizes.
\end{abstract} }

\begin{document}

\maketitle

\theabstract

\section{Introduction}

Data privacy is a critical concern in distributed applications
including privacy-preserving machine learning
\cite{li2021privacy,knott2021crypten,koch2020privacy,liu2020privacy}
and blockchains
\cite{ishai2009zero,lu2019honeybadgermpc,gao2022symmeproof,tomaz2020preserving}.
Secure multiparty computation (MPC) is a key software-based enabling
technology for these applications. MPC protocols provide both confidentiality
and integrity properties, formulated as real/ideal (aka simulator)
security and universal composability (UC) \cite{evans2018pragmatic},
with well-studied manual proof methods \cite{Lindell2017}.  MPC
security semantics has also been reformulated as
\emph{hyperproperties}
\cite{8429300,10.1145/3453483.3454074,skalka-near-ppdp24} to support
automated language-based proof methods. Recent research in the SMT
community has developed theories of finite fields \cite{SMFF} with
clear relevance to verification of cryptographic schemes that rely on
field arithmetic. The goal of this paper is to combine these two
approaches in a decidable type system for verifying correctness and
security properties of MPC protocols.

Our focus is on low-level MPC protocols. Previous high-level
MPC-enabled languages such as Wysteria \cite{rastogi2014wysteria} and
Viaduct \cite{10.1145/3453483.3454074} are designed to provide
effective programming of full applications, and incorporate
sophisticated compilation techniques such as orchestration
\cite{viaduct-UC} to guarantee high-level security properties. But
these approaches rely on \emph{libraries} of low-level MPC protocols,
such as binary and arithmetic circuits. These low-level protocols are
probabilistic, encapsulate abstractions such as secret sharing and
semi-homomorphic encryption, and are complementary to high-level
language design.

\subsection{Overview and Contributions}

Our work is based on the $\metaprot$/$\minicat$ framework developed in
previous work \cite{skalka-near-ppdp24}. The prior work developed the 
language and semantics, plus two mechanisms for verifying security: an
automated approach for protocols in $\mathbb{F}_2$ (the binary field),
based on enumerating all possible executions of the protocol, that scales 
only to very small protocols; and a manual approach based on a program logic.
In Section \ref{section-lang} we recall the $\minicat$ language model and key definitions. In
Section \ref{section-model} we recall formulations of program distributions,
and key probabilistic hyperproperties of MPC in $\minicat$. 

Together, our contributions \emph{automate} the verification of security
properties in $\metaprot$/$\minicat$. In this work, we build on~\cite{skalka-near-ppdp24}
by developing automatic verification approaches that work for arbitrary
finite fields and scale to large protocols. To accomplish this, we develop new
type systems for compositional verification, and use SMT to automate the checks
that were accomplished by enumeration in prior work. Specifically, we make the
following contributions:
\begin{enumerate}[i.]
\item In Section \ref{section-smt}, an interpretation of protocols in
  $\minicat$ as SMT constraints (Theorem \ref{theorem-toeq}) and
  methodology for verifying correctness properties.
\item In Section \ref{section-cpj}, a confidentiality type system for
  $\minicat$ protocols with a soundness property guaranteeing that
  information is released only through explicit declassifications
  (Theorem \ref{theorem-cpj}).
\item In Section \ref{section-ipj}, an integrity type system for
  $\minicat$ protocols with a soundness property guaranteeing
  robustness in malicious settings (Theorem \ref{theorem-ipj}).
\item In Section \ref{section-metalang}, a dependent Hoare type system and
  algorithm for the $\metaprot$ language that is sound for the
  confidentiality and integrity type systems (Theorem
  \ref{theorem-mtj}), with independently verified, compositional $\Pi$
  types for functions.
\end{enumerate}
In all cases, our type systems are boosted by verification mechanisms
provided by Satisfiability Modulo Finite Fields \cite{SMFF}. In
Section \ref{section-examples} we develop and discuss extended example
applications of our type analyses to real protocols including
the Goldreich-Micali-Wigderson (GMW), Bendlin-Damgard-Orland-Zakarias
(BDOZ), and Yao's Garbled Circuits (YGC) protocols.

\subsection{Related Work}
\label{section-related-work}

The goal of the $\minicat/\metaprot$ framework is to automate
correctness and security verification of low-level protocols.  Other
prior work has considered automated verification of high-level
protocols and manual verification of low-level protocols.
We summarize this comparison in Figure \ref{fig-comp-wrap}.

% We consider related systems in several
% dimensions, including whether they are aimed at low-level design with
% probabilistic features, whether they support reasoning about
% conditional probabilities which are central to real/ideal security,
% whether they consider MPC security through the lens of
% hyperproperties, and whether they consider passive and/or malicious
% security models.

The most closely related work is \cite{skalka-near-ppdp24}, where
foundations for $\minicat/\metaprot$ are introduced including a
verification algorithm. However, this algorithm only works for
$\mathbb{F}_2$-- e.g., binary circuits-- and does not scale due to
exponential complexity, necessitating semi-automated proof
techniques. In contrast, our type systems scale to arbitrary prime
fields and whole program analysis for larger circuits (protocols). So
our work is a significant advancement of $\minicat/\metaprot$.  Our
work is also closely related to Satisfiability Modulo Finite Fields
\cite{SMFF} but only in the sense that we apply their method as we
discuss in Section \ref{section-smt} and subsequently.

Other low-level languages with probabilistic features for privacy
preserving protocols have been proposed. The
$\lambda_{\mathrm{obliv}}$ language \cite{darais2019language} uses a
type system to automatically enforce so-called probabilistic trace
obliviousness.  But similar to previous work on oblivious data
structures \cite{10.1145/3498713}, obliviousness is related to pure
noninterference, not the conditional form related to passive MPC
security (Definition \ref{definition-NIMO}). The Haskell-based security
type system in \cite{6266151} enforces a version of noninterference
that is sound for passive security, but does not consider malicious
security. And properties of real/ideal passive and malicious security
for a probabilistic language have been (manually) formulated in
EasyCrypt \cite{8429300}.

\compwrapfig

Several high-level languages have been developed for writing MPC
applications. Previous work on analysis for the
SecreC language \cite{almeida2018enforcing,10.1145/2637113.2637119} is
concerned with properties of complex MPC circuits, in particular a
user-friendly specification and automated enforcement of
declassification bounds in programs that use MPC in subprograms. The
Wys$^\star$ language \cite{wysstar}, based on Wysteria
\cite{rastogi2014wysteria}, has similar goals and includes a
trace-based semantics for reasoning about the interactions of MPC
protocols. Their compiler also guarantees that underlying
multi-threaded protocols enforce the single-threaded source language
semantics. These two lines of work were focused on passive
security. The Viaduct language \cite{10.1145/3453483.3454074} has a
well-developed information flow type system that automatically
enforces both confidentiality and integrity through hyperproperties
such as robust declassification, in addition to rigorous compilation
guarantees through orchestration \cite{viaduct-UC}. However, these
high level languages lack probabilistic features and other
abstractions of low-level protocols, the implementation and security
of which are typically assumed as a selection of library components.

Program logics for probabilistic languages and specifically reasoning
about properties such as joint probabilistic independence is also
important related work. Probabilistic Separation Logic (PSL)
\cite{barthe2019probabilistic} develops a logical framework for
reasoning about probabilistic independence (aka separation) in
programs, and they consider several (hyper)properties, such as perfect
secrecy of one-time-pads and indistinguishability in secret sharing,
that are critical to MPC. Lilac \cite{li2023lilac} extends this line
of work with formalisms for conditional independence which is 
also particularly important for MPC \cite{skalka-near-ppdp24}.

The Cryptol and SAW tools \cite{10.1007/978-3-319-48869-1_5} allow
programmers to program security protocols that are verified using
SMT. However they are used for cryptographic protocols more generally,
and are not designed specifically for MPC as our framework is.

\section{$\minicat$ Syntax and Operational Semantics}
\label{section-lang}

The $\minifed$ language establishes a basic model of synchronous
protocols between a federation of \emph{clients} exchanging values in
the binary field. We identify clients by natural numbers and federations- finite sets of
clients- are always given statically.  Our threat model assumes a
partition of the federation into \emph{honest} $H$ and \emph{corrupt}
$C$ subsets. We model probabilistic programming via a \emph{random
tape} semantics. That is, we will assume that programs can make
reference to values chosen from a uniform random distributions defined
in the initial program memory.  Programs aka protocols execute
deterministically given the random tape.

\subsection{Syntax}

\minicatsyntaxfig

The syntax of $\minifed$, defined in Figure \ref{fig-minicat-syntax},
includes values $v$ and standard operations of addition, subtraction,
and multiplication in a finite field $\mathbb{F}_p$ where $p$ is some
prime.  Protocols are given input secret values $\secret{w}$ as well
as random samples $\flip{w}$ on the input tape, implemented using a
\emph{memory} as described below (Section
\ref{section-lang-semantics}) where $w$ is a distinguishing 
identifier string. Protocols are sequences of assignment commands of three
different forms:
\begin{itemize}
\item $\eassign{\mesg{w}}{\cid_2}{\be}{\cid_1}$: This
  is a \emph{message send} where expression $\be$ is computed
  by client $\cid_1$ and sent to client $\cid_2$ as message
  $\mesg{w}$.
\item $\reveal{w}{\be}{\cid}$: This
  is a \emph{public reveal} where expression $\be$ is computed
  by client $\cid$ and broadcast to the federation, typically
  to communicate intermediate results for use in final output
  computations.
\item $\pubout{\cid}{\be}{\cid}$: This
  is an \emph{output} where expression $\be$ is computed
  by client $\cid$ and reported as its output. As a
  sanity condition we disallow commands
  $\pubout{\cid_1}{\be}{\cid_2}$ where $\cid_1\ne\cid_2$.
\end{itemize}
The distinction between
messages and broadcast public reveal is consistent with previous
formulations, e.g., \cite{6266151}. To identify and distinguish
between collections of variables in protocols we introduce the
following notation.
\begin{definition}
We let $x$ range over \emph{variables} which are identifiers where
client ownership is specified- e.g.,
$\elab{\mesg{\mathit{foo}}}{\cid}$ is a message $\mathit{foo}$ that
was sent to $\cid$. We let $X$ range over sets of variables, and more
specifically, $S$ ranges over sets of secret variables
$\elab{\secret{w}}{\cid}$, $R$ ranges over sets of random variables
$\elab{\flip{w}}{\cid}$, $M$ ranges over sets of message variables
$\elab{\mesg{w}}{\cid}$, $P$ ranges over sets of public variables
$\rvl{w}$, and $O$ ranges over sets of output variables $\out{\cid}$.
Given a program $\prog$, we write $\iov(\prog)$ to denote the
particular set $S \cup M \cup P \cup O$ of variables in $\prog$ and
$\secrets(\prog)$ to denote $S$, and we write $\flips(\prog)$ to
denote the particular set $R$ of random samplings in $\prog$. We write
$\vars(\prog)$ to denote $\iov(\prog) \cup \flips(\prog)$. For any set
of variables $X$ and clients $I$, we write $X_I$ to denote the subset
of $X$ owned by any client $\cid \in I$, in particular we write $X_H$
and $X_C$ to denote the subsets belonging to honest and corrupt
parties, respectively.
\end{definition}
In all cases we disallow overwriting of variables, since assignment is
intended to represent distinct message sends. An example of a
$\minicat$ protocol, for additive secret sharing, is discussed in
Section \ref{section-smt-toeq}.

\subsection{Semantics}
\label{section-lang-semantics}

\minicatredxfig

\emph{Memories} are fundamental to the semantics of $\fedcat$ and
provide random tape and secret inputs to protocols, and also record
message sends, public broadcast, and client outputs.
\begin{definition}
Memories $\store$ are finite (partial) mappings from variables $x$ to
values $v \in \mathbb{F}_p$.  The \emph{domain} of a memory is written
$\dom(\store)$ and is the finite set of variables on which the memory
is defined.  We write $\store\{ x \mapsto v\}$ for
$x\not\in\dom(\store)$ to denote the memory $\store'$ such that
$\store'(x) = v$ and otherwise $\store'(y) = \store(y)$ for all $y \in
\dom(\store)$. We write $\store \subseteq \store'$ iff $\dom(\store)
\subseteq \dom(\store')$ and $\store(x) = \store'(x)$ for all $x \in
\dom(\store)$. Given any $\store$ and $\store'$ with $\store(x) =
\store'(x)$ for all $x \in \dom(\store) \cap \dom(\store')$, we write
$\store \uplus \store'$ to denote the memory with domain $X =
\dom(\store) \cup \dom(\store')$ such that:
$$
\forall x \in X .
(\store \uplus \store')(x) =
\begin{cases} \store(x) \text{\ if\ } x\in\dom(\store) \\ \store'(x) \text{\ otherwise\ }\end{cases} 
$$
%We write $\store \subseteq \store'$ iff $\store \uplus \store' = \store$.
\end{definition}
In our subsequent presentation we will often want to consider arbitrary
memories that range over particular variables and to restrict
memories to particular subsets of their domain:
\begin{definition}
  Given a set of variables $X$ and memory $\store$, we write
  $\store_X$ to denote the memory with $\dom(\store_X) = X$ and
  $\store_X(x) = \store(x)$ for all $x \in X$. We define $\mems(X)$ as
  the set of all memories with domain $X$:
  $$
  \mems(X) \defeq \{ \store \mid \dom(\store) = X \}
  $$
\end{definition}
So for example, the set of all random tapes for
a protocol $\prog$ is $\mems(\flips(\prog))$, and the memory $\store_{\secrets(\prog)}$
is $\store$ restricted to the secrets in $\prog$.
%We let $\stores$ range
%over sets of memories with the same domain, and abusing notation
%we write $\dom(\stores)$ to denote the common domain,
%and $\stores_X \defeq \{ \store_X | \store \in \stores \}$.

Given a variable-free expression $\be$, we write $\cod{\be}$ to denote
the standard interpretation of $\be$ in the arithmetic field
$\mathbb{F}_{p}$. With the introduction of variables to expressions,
we need to interpret variables with respect to a specific memory, and
all variables used in an expression must belong to a specified client.
Thus, we denote interpretation of expressions $\be$ computed on a
client $\cid$ as $\lcod{\store,\be}{\cid}$. This interpretation is
defined in Figure \ref{fig-minicat-redx}, along with the big-step
reduction relation $\redx$ to evaluate commands. Reduction is a mapping
from \emph{configurations} $(\store,\prog)$ to final stores
where all three command forms- message send, broadcast, and
output- are implemented as updates to the memory $\store$.

\paragraph{Oblivious Transfer} A passive secure oblivious transfer (OT) protocol
based on previous work \cite{barthe2019probabilistic} can be defined
in $\minicat$, albeit assuming some shared randomness. Alternatively,
a simple passive secure OT can be defined with the addition of public
key cryptography as a primitive. But given the diversity of approaches
to OT, we instead assume that OT is abstract with respect to its
implementation, and any use is of the form
$\mx{w}{\cid_1} := \elab{\OT{\be_1}{\cid_1}{\be_2}{\be_3}}{\cid_2}$ --
given a \emph{choice bit} $\be_1$ provided by a receiver $\cid_1$, the
sender $\cid_2$ sends either $\be_2$
or $\be_3$.  Critically, $\cid_2$ learns nothing about $\be_1$ and
$\cid_1$ learns nothing about the unselected value. We return to
the formal interpretation of OT in Section \ref{section-smt}.

\subsection{$\minicat$ Adversarial Semantics}
  
\minicataredxfig

In the malicious model we assume that corrupt clients are in
thrall to an adversary $\adversary$ that does not necessarily follow
the rules of the protocol.  We model this by positing a $\arewrite$
function which is given a corrupt memory $\store_C$ and expression
$\be$, and returns a rewritten expression that can be interpreted to
yield a corrupt input. We define the evaluation relation that
incorporates the adversary in Figure \ref{fig-minicat-aredx}.

A key technical distinction of the malicious setting is that it
typically incorporates ``abort''. Honest parties implement strategies
to detect rule-breaking-- aka \emph{cheating}-- by using, e.g.,
message authentication codes with semi-homomorphic properties as in
BDOZ/SPDZ \cite{10.1007/978-3-030-68869-1_3}. If cheating is detected,
the protocol is aborted. To model this, we extend $\minifed$ with an
\ttt{assert} command:
$$
    \begin{array}{rcl@{\hspace{2mm}}r}
      \prog &::=& \cdots \mid \elab{\assert{\be = \be}}{\cid}
    \end{array}
$$
As we will discuss, this form of asserting equality between expressions
is sufficient to capture MAC checking and support integrity in protocols.
The big-step adversarial semantics relation $\aredx$ is then defined
in Figure \ref{fig-minicat-aredx} as a mapping from configurations to
a store $\store$ or $\bot$.

\subsubsection{Protocol Runs}

Our semantics require that random tapes contain values for all program
values $\elab{\flip{w}}{\cid}$ sampled from a uniform distribution
over $\mathbb{F}_p$. Input memories also contain input secret values
and possibly other initial view elements as a result of
pre-processing, e.g., Beaver triples for efficient multiplication,
and/or MACed share distributions as in BDOZ/SPDZ
\cite{evans2018pragmatic,10.1007/978-3-030-68869-1_3}. We define
$\runs(\prog)$ as the set of final memories resulting from execution
of $\prog$ given any initial memory, possibly augmented with a
preprocessed memory (a memory that satisfies a predicate $\preproc$),
and $\aruns(\prog)$ is the analog in the malicious setting. Since we
constrain programs to not overwrite any variable, we are assured that
final memories contain both a complete record of all initial secrets
as well as views resulting from communicated information. This
formulation will be fundamental to our consideration of probability
distributions and hyperproperties of protocols.
%In this
%setting, given a program $\prog$ with $\iov(\prog) = S \cup V$ and
%$\flips(\prog) = F$ we will consider all $\store \in \mems(S \cup V
%\cup F)$ such that $ \config{\store_{S \cup F}}{\prog} \redxs
%\config{\store_}{\varnothing} $ to be equally probable, establishing
%the basic distribution of the program. %From this, we can immediately
%derive the marginal distribution of $S \cup V$ to reason about
%dependencies between secrets and views.
\begin{definition}
  \label{def-runs}
  Given $\prog$ with $\secrets(\prog) = S$ and $\flips(\prog) = R$ and
  predicate $\preproc$ on memories, define:
  $$
  \begin{array}{c}
    \runs(\prog) \defeq 
    \{ \store \mid \exists \store_1 \in \mems(S \cup R) . 
    \exists \store_2 . \preproc(\store_2) \wedge
    %(\dom(\store) = \iov(\prog) \cup R) \wedge
    (\store_1 \uplus \store_2,\prog) \redxs \store \}
  \end{array}
  $$
  And similarly for any $\adversary$ define:
  $$
  \begin{array}{c}
    \aruns(\prog) \defeq 
    \{ \store \mid \exists \store_1 \in \mems(S \cup R) . 
    \exists \store_2 . \preproc(\store_2) \wedge
    %(\dom(\store) = \iov(\prog) \cup R) \wedge
    (\store_1 \uplus \store_2,\prog) \aredxs \store \}
  \end{array}
  $$
\end{definition}

\section{Security Model}
\label{section-model}

MPC protocols are intended to implement some \emph{ideal
functionality} $\idealf$ with per-client outputs. In the $\minifed$
setting, given a protocol $\prog$ that implements $\idealf$, with
$\iov(\prog) = S \cup M \cup P \cup O$, the domain of $\idealf$ is
$\mems(S)$ and its range is $\mems(O)$. Security in the MPC setting
means that, given $\store \in \mems(S)$, a secure protocol $\prog$
does not reveal any more information about honest secrets $\store_H$
to parties in $C$ beyond what is implicitly declassified by
$\idealf(\sigma)$. Security comes in \emph{passive} and
\emph{malicious} flavors, wherein the adversary either follows the
rules or not, respectively. Characterization of both real world
protocol execution and simulation is defined probabilistically.  In
this work we will focus on the enforcement of hyperproperties of
passive and malicious security developed in related work
\cite{skalka-near-ppdp24}, that were shown to be sound for real/ideal
security properties that are traditionally used for MPC.  As in that
work we use probability mass functions to express joint dependencies
between input and output variables, as a metric of information
leakage.

\subsection{Probability Mass Functions (pmfs)} 

Pmfs map \emph{realizations} of variables in a joint
distribution to values in $[0..1]$. For convenience and following
\cite{barthe2019probabilistic,skalka-near-ppdp24} we use memories to
represent realizations, so for example given a pmf $\pmf$ over
variables $\{ \sx{x}{1}, \mx{y}{2} \}$ we write $\pmf(\{
\elab{\secret{x}}{1} \mapsto 0, \elab{\mesg{y}}{2} \mapsto 1 \})$ to
denote the (joint) probability that $\elab{\secret{x}}{1} = 0 \wedge
\elab{\mesg{y}}{2} = 1$. Recall from Section
\ref{section-lang-semantics} that $\uplus$ denotes the combination of
memories, so for example $\{ \elab{\secret{x}}{1} \mapsto 0\} \uplus
\{\elab{\mesg{y}}{2} \mapsto 1 \} = \{ \elab{\secret{x}}{1} \mapsto 0,
\elab{\mesg{y}}{2} \mapsto 1 \}$.
\begin{definition}
  A \emph{probability mass function} $\pmf$ is a function
  mapping memories in $\mems(X)$ for given variables $X$ to
  values in $\mathbb{R}$ such that:
  $$
  \sum_{\store \in \mems(X)} \pmf(\store) \  = \ 1
  $$
  %The \emph{support} of a distribution is the set of realizations
  %of random variables with non-zero probability:
  %$\support(\pmf) \defeq \{ (v_1,\ldots,v_n) \mid
  %\pmf(\{ x_1 \mapsto v_1, \ldots, x_n \mapsto v_n\}) > 0 \} $
\end{definition}
%To recover succinct and familiar notation, we may omit the domain of a
%distribution when it is clear from an application context-
%i.e., we allow the following sugaring:
%$$
%\pdf{}(\store) \defeq \pdf{\dom(\store)}(\store)
%$$
Now, we can define a notion of marginal and conditional
distributions as follows, which are standard for discrete
probability mass functions. 
\begin{definition}
  Given $\pmf$ with $\dom(\pmf) = \mems(X_2)$, the \emph{marginal distribution}
  of variables $X_1 \subseteq X_2$ in $\pmf$ is denoted $\margd{\pmf}{X_1}$ and defined as follows:
  $$
  \forall \store \in \mems(X_1) \quad . \quad \margd{\pmf}{X_1}(\store) \defeq
  \sum_{\store' \in \mems(X_2-X_1)} \pmf(\store \uplus \store')
  $$
\end{definition}

\begin{definition}
  Given $\pmf$, the \emph{conditional distribution}
  of $X_1$ given $X_2$ where $X_1 \cup X_2 \subseteq \dom(\pmf)$ and $X_1 \cap X_2 = \varnothing$
  is denoted $\condd{\pmf}{X_1}{X_2}$ and defined as follows:
  $$
  \forall \store \in \mems(X_1 \cup X_2)\ .\ 
  \condd{\pmf}{X_1}{X_2}(\store) \defeq
  \begin{cases}
    \begin{array}{ll}
      0 & \text{\ if\ } \margd{\pmf}{X_2}(\store_{X_2}) = 0\\
      \margd{\pmf}{X_1 \cup X_2}(\store) / \margd{\pmf}{X_2}(\store_{X_2})\ & \text{\ o.w.}
     \end{array}
  \end{cases}
  $$
\end{definition}
We also define some convenient syntactic sugarings. The first will allow us to
compare marginal distributions under different realization conditions
(as in, e.g., Definition \ref{definition-NIMO}), the others are standard.
\begin{definition}
  Given $\pmf$, for all $\store_1 \in \mems(X_1)$ and $\store_2 \in \mems(X_2)$ define:
  \begin{enumerate}
  \item $\condd{\pmf}{X_1}{\store_2}(\store_1) \defeq \condd{\pmf}{X_1}{X_2}(\store_1 \uplus \store_2)$
  \item $\pmf(\store_1)  \defeq \margd{\pmf}{X_1}(\store_1)$ 
  \item $\pmf(\store_1|\store_2) \defeq \condd{\pmf}{X_1}{X_2}(\store_1 \uplus \store_2)$
  \end{enumerate}
\end{definition}
As in previous work \cite{darais2019language,barthe2019probabilistic,skalka-near-ppdp24,li2023lilac}, 
probabilistic independence, aka separation, is an important concept and we adopt the
following standard notation to express conditional and unconditional separation:
\begin{definition}
%  We write $\vc{\pmf}{x}{y}$ iff $\pmf(\{ x \mapsto 0\}\ |\ \{ y \mapsto 0 \}) =
%  \pmf(\{ x \mapsto 1\}\ |\ \{ y \mapsto 1 \}) = 1$.
  We write $\condsep{\pmf}{X_1}{X_2}{X_3}$ iff for all
    $\store \in \mems(X_1 \cup X_2)$ and $\store' \in \mems(X_3)$ we have
  $\pmf(\store|\store') =
  \pmf(\store_{X_1}|\store') * \pmf(\store_{X_2}|\store')$. If $X_1$ is empty
  we write $\sep{\pmf}{X_2}{X_3}$.
\end{definition}

\subsection{Basic Distribution of a Protocol}

We treat all elements of $\runs(\prog)$ as equally likely.  This
establishes the basic program distribution that can be marginalized
and conditioned to quantify input/output information dependencies.
%In this
%setting, given a program $\prog$ with $\iov(\prog) = S \cup V$ and
%$\flips(\prog) = F$ we will consider all $\store \in \mems(S \cup V
%\cup F)$ such that $ \config{\store_{S \cup F}}{\prog} \redxs
%\config{\store_}{\varnothing} $ to be equally probable, establishing
%the basic distribution of the program. %From this, we can immediately
%derive the marginal distribution of $S \cup V$ to reason about
%dependencies between secrets and views.
\begin{definition}
  \label{def-progtt}
  \label{def-progd}
  \label{definition-progd}
  The \emph{basic distribution of $\prog$}, written $\progtt(\prog)$, is
  defined such that for all $\store \in \mems(\vars(\prog))$:
  $$
  \progtt(\prog)(\store) =  1 / |\runs(\prog)| \ \text{if}\ \store \in \runs(\prog), \text{otherwise}\ 0
  $$
  For any $\adversary$ the basic \emph{$\adversary$ distribution of $\prog$}, written
  $\progtt(\prog,\adversary)$, is
  defined such that for all $\store \in \mems(\iov(\prog) \cup R)$:
  $$
  \progtt(\prog,\adversary)(\store) =  1 / |\botruns(\prog)| \ \text{if}\ \store \in \botruns(\prog), \text{otherwise}\ 0
  $$
  where $\botruns$ pads out undefined views and outputs with $\bot$:
  $$
  \begin{array}{c}
    \botruns(\prog) \defeq \\
    \{ \store\{ x_1 \mapsto \bot, \ldots, x_n \mapsto \bot \} \mid 
    \store \in \aruns(\prog) \wedge \{ x_1,\ldots,x_n\} = (V \cup O) - \dom(\store) \}
  \end{array}
  $$
\end{definition}

\subsection{Honest and Corrupt Views}

Information about honest secrets can be revealed to corrupt clients
through messages sent from honest to corrupt clients, and through
publicly broadcast information from honest clients. Dually,
corrupt clients can impact protocol integrity through the messages
sent from corrupt to honest clients, and through publicly broadcast information
from corrupt clients. We call the former \emph{corrupt views}, and
the latter \emph{honest views}. We let $V$ range over sets
of views, i.e., subsets of $M \cup P$.
\begin{definition}[Corrupt and Honest Views]
  We let $V$ range over \emph{views} which are sets of messages
  and reveals. Given a program $\prog$ with $\iov(\prog) = S \cup M \cup P \cup O$,
  define $\views(\prog) \defeq M \cup P$, and define $\houtputs$ as
  the messages and reveals in $V = M \cup P$ sent from honest to corrupt
  parties, called \emph{corrupt views}:
  $$
  \begin{array}{lcl}
    \houtputs & \defeq
        & \{\ \rvl{w} \mid\ \reveal{w}{\be}{\cid} \in \prog \wedge \cid \in H \ \}\ \cup \\
      & & \{\ \elab{\mesg{w}}{\cid}\ \mid\  \eassign{\mesg{w}}{\cid}{\be}{\cid'} \in
           \prog \wedge \cid \in C \wedge \cid' \in H \ \} 
  \end{array}
  $$
  and similarly define $\cinputs$ as the subset of $V$ sent from corrupt to honest
  parties, called \emph{honest views}:
  $$
  \begin{array}{lcl}
    \cinputs &  \defeq
        & \{\ \rvl{w} \mid\ \reveal{w}{\be}{\cid} \in \prog \wedge \cid \in C \ \} \ \cup\\
      & & \{\ \elab{\mesg{w}}{\cid}\ \mid\  \eassign{\mesg{w}}{\cid}{\be}{\cid'} \in
              \prog \wedge \cid \in H \wedge \cid' \in C \ \}
  \end{array}
  $$
\end{definition}

\subsection{Hyperproperties of Confidentiality and Integrity}

In this section we restate hyperproperties developed in previous work
\cite{skalka-near-ppdp24}. The reader is referred to that paper for a more
thorough discussion of their relation to real/ideal security. In this
paper, we leverage that relation to establish security properties in
our type system.

\subsubsection{Confidentiality}

Since MPC protocols release some information about secrets through
outputs of $\idealf$, they do not enjoy strict noninterference.
Public reveals and protocol outputs are fundamentally forms of
declassification. But the following property of probabilistic
noninterference \emph{conditioned on output} is sound for passive
security \cite{skalka-near-ppdp24}.
\begin{definition}[Noninterference modulo output]
  \label{definition-NIMO}
  We say that a program $\prog$ with $\iov(\prog) = S \cup V \cup O$
  satisfies \emph{noninterference modulo output}
  iff for all $H$ and $C$ and $\store_1 \in \mems(S_C \cup O)$ and $\store_2 \in \mems(\houtputs)$
  we have:
  $$
  \condd{\progtt(\prog)}{S_H}{\store_1} = \condd{\progtt(\prog)}{S_H}{\store_1 \uplus \store_2}
 $$
\end{definition}
This conditional noninterference property implies that
corrupt views give the adversary no better chance of guessing honest
secrets than just the output and corrupt inputs do.

In practice, MPC protocols typically satisfy a \emph{gradual
release} property \cite{sabelfeld2009declassification}, where messages
exchanged remain probabilistically separable from secrets, with only
declassification events (reveals and outputs) releasing information
about honest secrets.  
\begin{definition}
  Given $H,C$, a protocol $\prog$ with $\iov(\prog) = S \cup M \cup P \cup O$
  satisfies \emph{gradual release} iff
  $\sep{\progtt(\prog)}{M_C}{S_H}$.
\end{definition}

\subsubsection{Integrity}

\emph{Integrity} is an important hyperproperty in security models that
admit malicious adversaries. Consistent with formulations in
deterministic settings, we define protocol integrity as the
preservation of high equivalence (of secrets and views). Intuitively,
this property says that any adversarial strategy either ``mimics'' a
passive strategy with some choice of inputs or causes an abort.
\begin{definition}[Integrity]
  \label{def-integrity}
  We say that a protocol $\prog$ with $\iov(\prog) = S \cup V \cup O$ has
  \emph{integrity} iff for all $H$, $C$, and $\adversary$,
  if $\store \in \aruns(\prog)$ 
  then there exists $\store' \in \mems(S)$ with $\store_{S_H} = \store'_{S_H} $ and:
    $$
    \condd{\progtt(\prog,\adversary)}{X}{\store_{S_H \cup \cinputs}} =
    \condd{\progtt(\prog)}{X}{\store'}
    $$ 
  where $X \defeq (\houtputs \cup O_H) \cap \dom(\store)$. 
\end{definition}
Integrity plus noninterference modulo output implies malicious security
in the real/ideal model \cite{skalka-near-ppdp24}, and integrity plus
gradual release is a probabilistic form of robust declassification
\cite{sabelfeld2009declassification}.

\section{$\minicat$ Constraint Verification}
\label{section-smt}

In previous work it was shown that the semantics of $\minicat$ in
$\mathbb{F}_2$ can be implemented with Datalog
\cite{skalka-near-ppdp24}, which is a form of constraint
programming. In this paper we extend this idea to arbitrary prime
fields by using a more general form of SMT constraint programming,
Satisfiability Modulo Finite Fields \cite{SMFF}. As we
will show in subsequent sections, this interpretation supports
correctness guarantees, and also static type analyses for enforcing
confidentiality and integrity properties.

\subsection{Constraint Satisfiability Modulo Finite Fields}

We introduce the following syntax of SMT-style constraints
for $\minicat$:
$$
\begin{array}{rcl}
  \phi ::= x \mid \phi \fplus \phi \mid \phi \fminus \phi \mid \phi \ftimes \phi  \qquad
  \eqs ::= \phi \eop \phi \mid \eqs \wedge \eqs \mid \neg\eqs
\end{array}
$$
%
%In Figure \ref{fig-eqs} we define a syntax of SMT-style constraints
%$\eqs$, where
Note that constraint terms $\phi$ are similar to expressions $\be$
except that $\phi$ can ``cross party lines''. This is needed
to express correctness properties-- for example, in the
Goldreich-Micali-Wigderson (GMW) protocol wire values in circuits are
represented by reconstructive shares \cite{evans2018pragmatic}.  If by
convention shares of values $n$ are represented by $\mesg{n}$ on
clients, then assuming two clients $\setit{1,2}$ the reconstructed
value can be expressed as $\mx{n}{1} + \mx{n}{2}$.  So, while summing
values across clients is disallowed in $\minicat$ protocols, we can
express properties of shared values via constraints.

%\begin{fpfig}[t]{Syntax of constraints}{fig-eqs}
%$$
%\begin{array}{rcl}
%  \phi ::= x \mid \phi \fplus \phi \mid \phi \fminus \phi \mid \phi \ftimes \phi  \qquad
%  \eqs ::= \phi \eop \phi \mid \eqs \wedge \eqs \mid \neg\eqs
%\end{array}
%$$
%
%\begin{mathpar}
%  \store(\phi_1 \fplus \phi_2) = \cod{\store(\phi_1) \fplus \store(\phi_2)}
%  
%  \store(\phi_1 \ftimes \phi_2) = \cod{\store(\phi_1) \ftimes \store(\phi_2)}
%  
%  \store(\phi_1 \fminus \phi_2) = \cod{\store(\phi_1) \fminus \store(\phi_2)}
%\end{mathpar}
%\end{fpfig}

This language of constraints has an obvious direct interpretation in
Satisfiability Modulo Finite Fields \cite{SMFF}. We can leverage this
to implement a critical \emph{entailment} property, written $\eqs_1
\models \eqs_2$.  Our entailment relation is based on satisfiability
in a standard sense, where we represent models as memories $\store$
(mappings from variables to values).
\begin{definition}
  We write $\store \models \eqs$ iff $\store$ satisfies, aka is a model
  of, $\eqs$. We write $\eqs_1 \models
  \eqs_2$ iff  $\store \models E_1$ implies $\store \models
  E_2$ for all $\store$, and note that this relation is a preorder.
\end{definition}
Given this Definition, the following Theorem is well-known and a fundamental
technique in SMT to implement our (common) notion of entailment. 
\begin{theorem}
  $\eqs_1 \models \eqs_2$ iff $\eqs_1 \wedge \neg\eqs_2$ is
  not satisfiable.
\end{theorem}

\subsection{Programs as Constraint Systems}
\label{section-smt-toeq}

\begin{fpfig}[t]{Interpretation of $\minicat$ expressions (T) and programs (B) as
  constraints}{fig-toeq}
\small{
\begin{mathpar}
  \toeq{x} = x

  \toeq{\elab{\be_1 \fplus \be_2}{\cid}} = (\toeq{\elab{\be_1}{\cid}} \fplus \toeq{\elab{\be_2}{\cid}})

  \toeq{\elab{\be_1 \fminus \be_2}{\cid}} = (\toeq{\elab{\be_1}{\cid}} \fminus \toeq{\elab{\be_2}{\cid}})

  \toeq{\elab{\be_1 \ftimes \be_2}{\cid}} = (\toeq{\elab{\be_1}{\cid}} \ftimes \toeq{\elab{\be_2}{\cid}})
\end{mathpar}
\begin{mathpar}
  \toeq{\xassign{x}{\be}{\cid}} = (x \eop \toeq{\elab{\be}{\cid}})
  
  \toeq{\elab{\assert{\be_1 = \be_2}}{\cid}} =  (\toeq{\elab{\be_1}{\cid}} \eop \toeq{\elab{\be_2}{\cid}})

  \toeq{\prog_1;\prog_2} = (\toeq{\prog_1} \wedge \toeq{\prog_2})
\end{mathpar}}
\end{fpfig}

A central idea of our approach is that we can interpret any protocol
$\prog$ as a set of equality constraints (denoted $\toeq{\prog}$) and use an SMT
solver to verify properties relevant to correctness, confidentiality,
and integrity. Further, we can leverage entailment for tractability--
we can use annotations to obtain a weakened precondition for relevant properties.
That is, given $\prog$, program annotations or other cues can be used
to find a minimal $\eqs$ with $\toeq{\prog} \models \eqs$ for verifying
correctness and security.

The mapping $\toeq{\cdot}$ from programs $\prog$ to constraints is
defined in Figure \ref{fig-toeq}. The interpretation of OT is omitted
from this figure which is general. For $\mathbb{F}_2$ (where we allow
negation $\neg$ of expressions) the interpretation is:
\begin{mathpar}
  \toeq{\elab{\OT{\be_1}{\cid_1}{\be_2}{\be_3}}{\cid_2}} =
  (\toeq{\elab{\be_1}{\cid_1}} \ftimes \toeq{\elab{\be_3}{\cid_2}}) \fplus
  (\neg\toeq{\elab{\be_1}{\cid_1}} \ftimes \toeq{\elab{\be_2}{\cid_2}}) 
\end{mathpar}
In general we will assume that any preprocessing predicate is defined
via a constraint $\eqspre$, i.e., $\preproc(\store) \iff \store
\models \eqspre$ for all $\store$.  The correctness of the
SMT interpretation of programs is characterized as follows.
\begin{theorem}
  \label{theorem-toeq}
  $\store$ is a model of $\eqspre \wedge \toeq{\prog}$ iff $\store \in \runs(\prog)$.
\end{theorem}

\subsubsection{Example: Correctness of 3-Party Addition}
Consider the following 3-party additive secret sharing protocol, where
each party generates two additive secret shares
\cite{evans2018pragmatic} for their respective secrets. Each party
publicly reveals the sums all their shares, and the sum of those is
the protocol output for each party.  {\footnotesize
$$
\begin{array}{c@{\qquad}c}
\begin{array}{lll}
  \elab{\mesg{s1}}{2} &:=& \elab{(\secret{1} \fminus \locflip \fminus \flip{x})}{1} \\ 
  \elab{\mesg{s1}}{3} &:=& \elab{\flip{x}}{1} \\ 
  \elab{\mesg{s2}}{1} &:=& \elab{(\secret{2} \fminus \locflip \fminus \flip{x})}{2} \\ 
  \elab{\mesg{s2}}{3} &:=& \elab{\flip{x}}{2} \\ 
  \elab{\mesg{s3}}{1} &:=& \elab{(\secret{3} \fminus \locflip \fminus \flip{x})}{3} \\ 
  \elab{\mesg{s3}}{2} &:=& \elab{\flip{x}}{3}
\end{array} & 
\begin{array}{lll}
  \rvl{1} &:=& \elab{(\locflip \fplus \mesg{s2} \fplus \mesg{s3})}{1} \\ 
  \rvl{2} &:=& \elab{(\mesg{s1} \fplus \locflip \fplus \mesg{s3})}{2} \\
  \rvl{3} &:=& \elab{(\mesg{s1} \fplus \mesg{s2} \fplus \locflip)}{3} \\
  \out{1} &:=& \elab{(\rvl{1} \fplus \rvl{2} + \rvl{3})}{1}\\
  \out{2} &:=& \elab{(\rvl{1} \fplus \rvl{2} + \rvl{3})}{2}\\
  \out{3} &:=& \elab{(\rvl{1} \fplus \rvl{2} + \rvl{3})}{3}
\end{array}
\end{array}
$$}
Letting $\prog$ be this protocol, we can verify correctness
as follows-- that is, the output sum of all shares is the sum of the
three input secrets:
$$
\toeq{\prog}\ \models\ \out{1} \eop \sx{1}{1} \fplus \sx{2}{2} \fplus \sx{3}{3}
$$

\section{Confidentiality Types}
\label{section-cpj}

It is well-known that adding or subtracting a sample from a uniform
distribution in a finite field yields a value in a uniform
distribution, meaning that samples can be used as one-time-pads with
perfect secrecy \cite{barthe2019probabilistic,darais2019language}.
In our type system for confidentiality, we aim to approximate
distributions by tracking which program variables an expression
may depend on. But we also want to capture this fundamental mechanism
of encryption. As we will show, SMT can be used in this and
other contexts to support confidentiality properties of protocols.

For example, given some string $y$ and a message send such as:
$$
\xassign{\mx{w}{1}}{\secret{w} \ftimes \flip{w}}{2}
$$
we would assign the following type to denote that $\mx{w}{1}$ is dependent on
both $\sx{w}{2}$ and $\rx{w}{2}$:
$$
\mx{w}{1} : \setit{\sx{w}{2}, \rx{w}{2}}
$$
But in the case of 2-party reconstructive sharing:
$$
\xassign{\mx{w}{1}}{\secret{w} \fminus \flip{w}}{2}
$$
we want the type of $\mx{w}{1}$ to be an ``encrypted'' type that is
by itself independent of $\sx{w}{2}$. However, we also want to track
the \emph{joint} dependence of $\mx{w}{2}$ on $\rx{w}{2}$ and
$\sx{w}{2}$, in case information about $\rx{w}{2}$ is subsequently
leaked.  In short, we assign $\mx{w}{1}$ a ``ciphertext type'', aka a
\emph{ciphertype}:
$$
\mx{w}{1} : \cty{\rx{w}{2}}{\setit{\sx{w}{2}}}
$$
indicating syntactically that $\mx{w}{1}$ is an encryption of a
value dependent on $\sx{w}{2}$ with $\rx{w}{2}$.

Of course, there are other methods for encrypting values in MPC
protocols-- but we can observe that many are algebraically equivalent
to this fundamental schema. For example, in Yao's Garbled Circuits
(YGC) in $\mathbb{F}_{2}$, the ``garbler'' represents encrypted wire
values as a random sample (for 1) or its negation (for 0), and shares
its secret input encoding with the ``evaluator'' by using
it as a selection bit. That is, assuming that client 2 is the garbler and
client 1 is the evaluator, we can define:
$$
\xassign{\mx{w}{1}}{\mux{\secret{w}}{\neg\flip{w}}{\flip{w}}}{2}
$$
where for all $\be_1,\be_2,\be_3$ (with $\neg$ denoting negation in $\mathbb{F}_2$):
$$
\mux{\be_1}{\be_2}{\be_3} \defeq (\neg\be_1 \ftimes \be_2) \fplus (\be_1 \ftimes \be_3)
$$
and letting this protocol be $\prog$ the following is valid:
\begin{mathpar}
%\toeq{\xassign{\mx{w}{1}}{\OT{\secret{w}}{1}{\neg\flip{w}}{\flip{w}}}{2}}
   \toeq{\prog} \models \mx{w}{1} \eop \neg\sx{w}{2} \fplus \rx{w}{2}
\end{mathpar}
This allows us to assign the same ciphertype to $\mx{w}{1}$ as above--
that is, $\cty{\rx{w}{2}}{\setit{\sx{w}{2}}}$.  Since this is a
one-time-pad scheme, it is also important that we ensure that samples
are used at most once for encryption. As in previous work 
\cite{darais2019language} we use type linearity to enforce this.

\cpjfig

The syntax of types $\ty$ and type environments $\Gamma$ are presented
in Figure \ref{fig-cpj}\footnote{In general we use $2^{\chi}$ to
denote the powerset of terms in any grammatical sort $\chi$.}. We also
define type judgements for expressions and programs by direct
interpretation as terms $\phi$ and constraints $\eqs$ respectively,
via the encoding $\toeq{\cdot}$. The \TirName{Encode} rule enforces
linearity by requiring that variables used for encryption are added to
judgements, and guaranteed to not be used elsewhere via a form of
union that ensures disjointness of combined sets:
\begin{definition}
  $R_1;R_2 = R_1 \cup R_2$ if $R_1 \cap R_2 = \varnothing$ and is otherwise
  undefined.
\end{definition}
The \TirName{DepTy} rule handles all cases and just records the
dependencies on variables in the term, a conservative approximation of
dependencies.  The \TirName{Send} rule records the types of messages,
reveals, and outputs in environments, and the \TirName{Seq} rule
combines the types of sequenced programs.  Given this we define
validity of program type judgements as follows.
\begin{definition}
  Given preprocessing predicate $\eqspre$ and protocol $\prog$ we say
  $\cpj{R}{\eqs}{\eqspre \wedge \toeq{\prog}}{\Gamma}$ is \emph{valid} iff it is derivable and
  $\eqspre \wedge \toeq{\prog} \models \eqs$.
\end{definition}
It is important to note that this presentation of derivations is
logical, not algorithmic-- in particular it is not syntax-directed
due to the \TirName{Encode} rule.

\subsection{Detecting Leakage Through Dependencies}

Given a valid typing, we can leverage structural type information to
conservatively approximate the information accessible by the
adversary who has control over all corrupt clients in $C$ given
a partitioning $H,C$. In particular, if the adversary has access to an
encrypted value and another value that is dependent on its one-time-pad
key-- i.e., access to values with types of the forms:
$$
\setit{\cty{\rx{w}{\cid}}{\ty}} \qquad \setit{\rx{w}{\cid}, \ldots} 
$$
then we can conservatively assume that the adversary can decrypt the
value of type $\ty$ and thus have access to those values.

In Figure \ref{fig-leakj} we define derivation rules for the
``leakage'' judgement $\leakj{\Gamma}{C}{\ty}$ which intuitively says
that given a protocol with type environment $\Gamma$ and corrupt
clients $C$, the adversary may be able to observe dependencies on
values of type $\ty$ in their adversarial messages.  The real aim of
this relation is to determine if
$\leakj{\Gamma}{C}{\setit{\sx{w}{\cid}}}$ with $\cid \in H$-- that is,
if the adversary may observe dependencies on honest secrets.

Our main type soundness result combines typing and leakage judgements,
and guarantees gradual release of protocols-- that is, that adversarial
messages reflect no dependencies on honest secrets. Proof details are
provided in an Appendix.
\begin{theorem}
  \label{theorem-cpj}
  Given $\prog$ with $\views(\prog) = M \cup P$, if $\cpj{R}{\eqs}{\eqspre \wedge \toeq{\prog}}{\Gamma}$
  is valid and for all $H,C$ we have $\leakj{\Gamma}{M_C}{T}$ with $T$ closed
  and $S_H \cap T = \varnothing$, then $\prog$ satisfies gradual release.
\end{theorem}

\leakjfig

It is important to recall from \cite{skalka-near-ppdp24} that gradual
release does \emph{not} imply passive security, since it only concerns
adversarial messages, but not public reveals of outputs where
dependencies on input secrets can occur even in secure
protocols. Nevertheless, gradual release is an important
confidentiality property for MPC (and other distributed) protocols.

\begin{comment}
\subsection{Examples}

\begin{verbatimtab}
m[s1]@2 := (s[1] - r[local] - r[x])@1
m[s1]@3 := r[x]@1

// m[s1]@2 : { c(r[x]@1, { c(r[local]@1, {s[1]@1} ) }
// m[s1]@3 : { r[x]@1 }
\end{verbatimtab}

\begin{verbatimtab}
m[x]@1 := s2(s[x],-r[x],r[x])@2

// m[x]@1 == s[x]@2 + -r[x]@2 
// m[x]@1 : { c(r[x]@2, { s[x]@2 }) } 

m[y]@1 := OT(s[y]@1,-r[y],r[y])@2

// m[y]@1 == s[y]@1 + -r[y]@2
// m[y]@1 : { c(r[y]@2, { s[y]@1 }) } 
\end{verbatimtab}
\end{comment}

\section{Integrity Types}
\label{section-ipj}

Our integrity type system is essentially a taint analysis, with
machinery to allow consideration of any $H,C$ partitioning.  We also
include rules to consider MAC codes to prevent cheating. We consider a
semi-homomorphic encryption scheme that is used in BDOZ and SPDZ
\cite{SPDZ1,SPDZ2,BDOZ,10.1007/978-3-030-68869-1_3} and is applicable
to various species of circuits. However other MAC schemes could be
considered with a modified derivation rule incorporated with our taint
analysis framework. BDOZ/SPDZ assumes a pre-processing phase that
distributes shares with MAC codes to clients.  This integrates well
with pre-processed Beaver Triples to implement malicious secure, and
relatively efficient, multiplication \cite{evans2018pragmatic}.

In BDOZ, a field value $v$ is secret shared among 2 clients 
with accompanying MAC values.  Each client $\cid$ gets a pair of
values $(v_\cid,m_\cid)$, where $v_\cid$ is a share of $v$
reconstructed by addition, i.e., $v = \fcod{v_1 \fplus v_2}$, and
$m_\cid$ is a MAC of $v_\cid$.  More precisely, $m_\cid = k +
(k_\Delta * v_\cid)$ where \emph{keys} $k$ and $k_\Delta$ are known
only to $\cid' \ne \cid$ and $k_\Delta$. The \emph{local key} $k$ is
unique per MAC, while the \emph{global key} $k_\Delta$ is common to
all MACs authenticated by $\cid'$. This is a semi-homomorphic
encryption scheme that supports addition of shares and multiplication
of shares by a constant-- for more details the reader is referred to
Orsini \cite{10.1007/978-3-030-68869-1_3}. We note that while we
restrict values $v$, $m$, and $k$ to the same field in this
presentation for simplicity, in general $m$ and $k$ can be in
extensions of $\mathbb{F}_p$.

\ipjfig

As in \cite{skalka-near-ppdp24}, we use $\macgv{\mesg{w}}$ to refer to
secret-shared values reconstructed with addition, and by
convention each share of $\macgv{\mesg{w}}$ is represented as
$\elab{\mesg{w\ttt{s}}}{\cid}$, the MAC of which is represented as a
$\elab{\mesg{w\ttt{m}}}{\cid}$ for all $\cid$, and each client holds a
key $\elab{\mesg{w\ttt{k}}}{\cid}$ for authentication of the other's
share. Each client also holds a global key $\mesg{\ttt{delta}}$. For
any such share identified by string $w$, the BDOZ MAC scheme is defined 
by the equality predicate $\macbdoz{w}$:
$$
\macbdoz{w} \defeq
\mesg{w\ttt{m}} = \mesg{w\ttt{k}} \fplus \ttt{(}\mesg{\ttt{delta}} \ftimes
\mesg{w\ttt{s}}\ttt{)}
$$

\subsection{Typing and Integrity Labeling}

The integrity of values during protocol executions depends on the
partitioning $H,C$, since any value received from a member of $C$ is
initially considered to be low integrity. But as for confidentiality
types, we want our basic type analysis to be generalizable to
arbitrary $H,C$. So integrity type judgements, defined in Figure
\ref{fig-ipj}, are of the form $\ipj{\eqs}{\prog}{\Delta}$, where
$\eqs$ plays a similar role as in confidentiality types, as a possibly
weakened expression of algebraic properties of $\prog$, and
environments $\Delta$ record dependency information that can be
specialized to a given $H,C$. Integrity environments are ordered lists
since order of evaluation is important for tracking integrity--
validating low integrity information through MAC checking allows
subsequent treatment of it as high integrity, but not before. Note
that MAC checking is \emph{not} a form of endorsement as it is usually
defined (as a dual of declassification
\cite{sabelfeld2009declassification}), since a successful check
guarantees no deviation from the protocol.

The typing of protocols depends on a type judgement for expressions
of the form $\itj{\cid}{\be}{V}$, which records
messages and reveals $V$ used in the construction of $\be$.
This source is recorded in $\Delta$ for message sends-- so for
example, for some string $w$ letting:
$$
\prog\ \defeq\ (\xassign{\mx{w}{1}}{\secret{w} - \flip{w}}{2};\ \rvl{w} := \mx{w}{1})
$$
the following judgement is valid:
$$
\ipj{\toeq{\prog}}{\prog}{(\mx{w}{1} : \ity{2}{\varnothing}; \rvl{w} : \ity{1}{\mx{w}{1}})}
$$
More formally, we define validity as follows:
\begin{definition}
  Given pre-processing predicate $\eqspre$ and protocol $\prog$, 
  we say $\ipj{\eqs}{\prog}{\Delta}$ is \emph{valid}
  iff it is derivable and $\eqspre \wedge \toeq{\prog} \models E$.
\end{definition}

In the above example, note that the environment implicitly records
that $\rvl{w}$ depends on data from client 1, which in turn
depends on data from client 2. The rule for MAC
checking allows an overwrite of the dependency. We explore
the application of this in more detail in Section
\ref{section-examples}.

\cheatjfig

\subsection{Assigning Integrity Labels}

In all cases, we want to determine the integrity level of computed
values based on assumed partitionings $H,C$. In the previous example,
we want to say that $\rvl{w}$ has low integrity if either client 1 or
2 is corrupt.  To this end we define \emph{security labelings}
$\seclev$ which are mappings from variables $x$ to security levels
$\latel$. By default, given a partitioning $H,C$ this labeling maps
messages to the integrity level of the owner.
\begin{definition}  
  Given $H,C$,
  define $\seclev_{H,C}$ such that $\seclev_{H,C}(\mx{w}{\cid} ) = \hilab$
  if $\cid \in H$  and $\lolab$
  otherwise for all $\mx{w}{\cid}$.
\end{definition}
But the values of messages and reveals may be affected by values sent
by other parties, which is recorded in the $\Delta$ obtained from its
typing. To determine whether this impacts integrity for any given
$H,C$, we apply the inductively defined rewrite relation
$\cheatj{\Delta}{H,C}{\seclev}$, axiomatized in Figure
\ref{fig-cheatj}. The composition of typing and this latter rewriting
obtains essentially a traditional taint analysis, modulo MAC checking.
Our main type soundness result for integrity is then given as follows.
Proof details are given in the Appendix.
\begin{theorem}
  \label{theorem-ipj}
  Given  $\eqspre$ and $\prog$ with
  $\iov(\prog) = S \cup V \cup O$, if
  $\ipj{\eqs}{\prog}{\Delta}$ is valid
  and for all $H,C$ with $\cheatj{\Delta}{H,C}{\seclev}$ 
  we have $\seclev(x) = \hilab$ for all $x \in \houtputs \cup O_H$, then $\prog$
  has integrity for all $H,C$.
\end{theorem}

\section{Compositional Type Verification in $\metaprot$}
\label{section-metalang}

The $\metaprot$ language \cite{skalka-near-ppdp24} includes structured
data and function definitions for defining composable protocol
elements at a higher level of abstraction than $\minicat$.  The
$\metaprot$ language is a \emph{metalanguage} aka metaprogramming
language, where an $\fedprot$ protocol is the result of
computation. In addition to these declarative benefits of $\metaprot$,
component definitions support compositional verification of larger
protocols. Separate verification of well-designed components results
in confidentiality and integrity properties to be recorded in their
types, allowing for significant reduction of SMT verification in whole
program analysis, as we will discuss with extended examples in Section
\ref{section-examples}.

\metaprotsyntaxfig

\subsection{Syntax and Semantics}

The syntax of $\metaprot$ is defined in Figure
\ref{fig-metaprotsyntax}.  It includes a syntax of values $\mv$ that
include client ids $\cid$, identifier strings $w$, expressions $\be$
in field $\mathbb{F}_p$, record values, and $\minicat$ variables
$x$. $\metaprot$ expression forms allow dynamic construction of these
values. $\metaprot$ \emph{instruction} forms allow dynamic
construction of $\minicat$ protocols $\prog$ that incorporate expression
evaluation. The syntax also supports definitions of functions that
compute values $\mv$ and, as a distinct form, functions that compute
protocols $\prog$.  Formally, we consider a complete metaprogram to include
both a codebase and a ``main'' program that uses the codebase.
We disallow recursion, mainly to guarantee decidability
of type dependence (Section \ref{section-metalangty}).
\begin{definition}
A \emph{codebase} $\codebase$ is a list of function 
declarations. We write $ \codebase(f) = y_1,\ldots,y_n,\ e$
if $f(y_1,\ldots,y_n) \{ e \} \in \codebase$, and we
write  $ \codebase(f) = y_1,\ldots,y_n,\ \instr$
if $f(y_1,\ldots,y_n) \{ \instr \} \in \codebase$.
%A \emph{metaprogram}, aka \emph{metaprotocol}  is a pair of a 
%codebase and expression $\codebase, e$. We may omit
%$\codebase$ if it is clear from context.  
\end{definition}

\metaprotexprsemanticsfig

We define a big-step evaluation relation $\redx$ in Figures
\ref{fig-metaprotexprsemantics} and \ref{fig-metaprotinstrsemantics}
for expressions and instructions, respectively.  In this definition we
write $e[\mv/y]$ and $\instr[\mv/y]$ to denote the substitution of
$\mv$ for free occurrences of $y$ in $e$ or $\instr$ respectively. The
rules are mostly standard. Note that we do not include an evaluation
rule for the form $\eqcast{\mx{e}{e}}{\notg{\phi}}$ which is a type
annotation that we assume is erased from programs prior to
evaluation. We defer discussion of this form, as well as the syntactic
category $\notg{\phi}$, to the next Section.

\metaprotinstrsemanticsfig

\subsection{Dependent Hoare Type Theory}
\label{section-metalangty}

Our first goal in the $\metaprot$ type theory is to define an
algorithmic system that is sound for both confidentiality and
integrity typings as defined in Sections \ref{section-cpj} and
\ref{section-ipj}. Returning to the YGC secret encoding example in
Section \ref{section-cpj}, note that the key equivalence of the
$\ttt{mux}$ expression with a one-time-pad encryption we observe there
is not trivial. To ensure that this sort of equivalence can be picked
up by the type system, we introduce an $\ttt{as}$ annotation form that
allows the programmer to provide the needed hint.

Consider the following $\ttt{encode}$ function definition that
generalizes this YGC encoding for any identifier $y$ and sender, receiver
pair $s,r$. The second line in the body is a hint that
the message $\mx{y}{r}$ can be considered equivalent to
$\neg\sx{y}{s} \fplus \rx{y}{s}$:
$$
\begin{array}{l}
\ttt{encode}(y,s,r) \{\\
\quad \xassign{\mx{y}{r}}{\mux{\secret{y}}{\neg\flip{y}}{\flip{y}}}{s};\\
\quad \eqcast{\mx{y}{r}}{\neg\sx{y}{s} \fplus \rx{y}{s}}\\
\}
\end{array}
$$
This hint can be validated using SMT, and then the
syntactic structure of $\neg\sx{y}{s} \fplus \rx{y}{s}$
allows its immediate interpretation as a one-time-pad encoding.

Our second goal in the $\metaprot$ type system is to minimize the
amount of SMT solving needed for type verification.  Returning to the
$\ttt{encode}$ example, type checking in our systems only requires the
hint to be verified once, with the guarantee that $\ttt{encode}$ can
be applied anywhere without needing to re-verify the hint in
application instances. To verify the hint in $\ttt{encode}$, we can
just choose arbitrary ``fresh'' values $w,\cid_1,\cid_2$ for $y$, $s$,
and $r$, evaluate $\ttt{encode}(w,\cid_1,\cid_2) \redx \prog$, and
then verify:
$$\toeq{\prog} \models \mx{w}{\cid_2} \eop
\neg\sx{w}{\cid_1} \fplus \rx{w}{\cid_1}$$ 
Since $\ttt{encode}$ is closed, validity is guaranteed for
any instance of $y$, $s$, and $r$.

We generalize these benefits of compositional verification by allowing
pre-~and postcondition annotations on $\metaprot$ functions. For
example, consider a GMW-style, 2-party ``and-gate'' function
$\ttt{andgmw}(x,y,z)$. In this protocol, each party $\cid
\in \{1,2\}$ holds an additive secret share $\mx{x}{\cid}$ of
values identified by $x$ and $y$, and at the end of the
protocol each hold a secret share $\mx{z}{\cid}$, where:
$$
\mx{z}{1} \fplus \mx{z}{2} \eop (\mx{x}{1} \fplus \mx{x}{2}) \ftimes (\mx{y}{1} \fplus \mx{y}{2})
$$
We provide details of $\ttt{andgmw}$ and other GMW protocol
elements in Section \ref{section-examples}.  Our point now is that,
similarly to the $\ttt{encode}$ example, we can verify this
postcondition once as a correctness property for $\ttt{andgmw}$, and
then integrate instances of it into circuit correctness properties
with the guarantee that each instance also holds for any
$\ttt{andgmw}$ gate.  Since the program logic of $\ttt{andgmw}$ is
non-trivial, and typical circuits can use up to thousands of gates,
this has significant practical benefits by greatly reducing SMT overhead
in whole-program analysis.

\atjfig

\subsubsection{$\minicat$ expression type algorithm.}

A core element of $\metaprot$ type checking is type checking
of $\minicat$ expressions. The integrity type system presented in Section
\ref{section-ipj} is already algorithmic and ready to use. But confidentiality
typing presented in Section \ref{section-cpj} is not syntax-directed
due to the \TirName{Encode} rule. But as described above, by introducing
hint annotations we can ``coerce'' any relevant expression form into
the simplest equivalent syntactic form for one-time-pad encoding.
Thus, in Figure \ref{fig-atj} we present the algorithm for
confidentiality type judgements, where we eliminate the need for
integrated SMT solving by assuming this sort of casting. 

The following Lemma establishes correctness of algorithmic confidentiality
type checking, and makes explicit that SMT checking is eliminated in the
judgement.
\begin{lemma}
  \label{lemma-atj-sound}
  If $\atj{\toeq{\phi}}{R}{\ty}$ then $\eqj{R}{\eqs}{\phi}{\ty}$ for any $\eqs$.
\end{lemma}

\subsubsection{Dependent Hoare types for instructions.}

Hoare-style types for instructions have the following form:
$$
\hty{\eqs_1}{\Gamma}{R}{\Delta}{\eqs_2}
$$
Here, $\eqs_1$ and $\eqs_2$ are the pre- and postconditions
respectively, $\Gamma$ is the confidentiality type environment of the
protocol resulting from execution of the instruction, which is sound
wrt confidentiality typing, $R$ are the one-time-pads consumed in the
confidentiality typing, and $\Delta$ is the sound integrity type
environment.

\notgfig

To obtain a more complete type system, and to generalize
function typings, we also introduce a form of type dependency,
specifically dependence on $\metaprot$ expressions. Dependent
$\Pi$ types have the form:
$$
\dht{y_1,\ldots,y_n}{\notg{\eqs_1}}{\notg{\Gamma}}{\notg{R}}{\notg{\Delta}}{\notg{\eqs_2}}
$$
where $y_1,\ldots,y_n$ range over values $\mv$ and each of
$\notg{\eqs_1}$ etc.~may contain expressions with free occurrences
of $y_1,\ldots,y_n$-- the syntax of these forms is in Figure
\ref{fig-notg}. These $\Pi$ types are assigned to functions
and instantiated at application points. 

For example, here is a valid typing of $\ttt{encode}$.  It says that
under any precondition, evaluating $\ttt{encode}$ results in a cipher
type for the encoded message $\mx{x}{r}$, which consumes the
one-time-pad $\rx{x}{s}$, and the integrity of $\mx{x}{r}$ is
determined by the security level of $s$. The
postcondition expresses the key confidentiality property of
$\mx{x}{r}$, but also may be practically useful for correctness
properties since it is a simpler expression form than the $\ttt{mux}$
form:
$$
\ttt{encode} :
\begin{array}[t]{ll}
  \Pi x,s,r . & \{ \}\quad \mx{x}{r} : \cty{\rx{x}{s}}{\sx{x}{s}}, \{ \rx{x}{s} \} \ \cdot\\
  & \phantom{\{ \}} \quad \mx{x}{r} : \ity{s}{\setit{}} \\
  & \{ \mx{x}{r} \eop \neg\sx{y}{s} \fplus \rx{y}{s} \}
\end{array}
$$
A key property of this example and our type system generally is that once
the $\Pi$ type is verified, typing any application of it requires verification
of the precondition instance, but not the postcondition instance.

\subsection{Algorithmic Type Validity in $\metaprot$}

\mtjfig

We equate types up to evaluation of subexpressions as defined in
Figure \ref{fig-notg}. Since expression evaluation is total we can
just evaluate types of closed instructions, which are guaranteed to
have closed types, to obtain Hoare typings for instructions. In Figure
\ref{fig-mtj} we define the type derivation rules for $\metaprot$. In
the \TirName{Mesg} rule we use algorithmic confidentiality and
integrity type checking in a straightforward manner.  In the
\TirName{Encode} rule, we verify the hint given in the type
annotation, and then use it for type checking.  In the \TirName{App}
rule we verify preconditions of the given function type and
instantiate the type and postconditions with the given expression
arguments.  The \TirName{Sig} rule we specify how to verify $\Pi$
types for functions. Note that the values $\mv_1,\ldots,\mv_n$ chosen
to instantiate the function parameters are chosen ``fresh'', meaning
that any strings and client identifiers occurring as substrings in
$\mv_1,\ldots,\mv_n$ are unique and not used in any program source
code. This ensures generality of verification as demonstrated in Lemma
\ref{lemma-eqs-notg} in the Appendix.

\mtjfnfig

We allow the programmer to specify a type signature $\tsig$ that
maps functions in $\codebase$ to valid dependent Hoare
typings. Define:
\begin{definition}
  $\tsig$ is \emph{verified} iff $f : \tsig(f)$ is valid for all $f \in \dom(\tsig)$.
\end{definition}
In practice, to define $\tsig$ the programmer only needs to specify
pre-~and postconditions, as we illustrate in Section
\ref{section-examples}, since the rest of the type is reconstructed by
type checking.

The following Theorem establishes our main result, that the
$\metaprot$ type checking is sound with respect to both
confidentiality and integrity typings. Since these typings imply
security hyperproperties of confidentiality and integrity, $\metaprot$
type checking implicitly enforces them. Proof details are given in the
Appendix.
\begin{theorem}[Soundness of $\metaprot$ type checking]
  \label{theorem-mtj}
  Given preprocessing predicate $\eqspre$, program $\instr$, and verified $\tsig$, if
  the judgement: $$\mtj{\instr}{\eqspre}{\Gamma}{R}{\Delta}{\eqs}$$ is derivable then
  $\instr \redx \prog$ and:
  \begin{enumerate}
  \item $\cpj{R}{\eqs}{\eqspre \wedge \toeq{\prog}}{\Gamma}$ is valid.
  \item $\ipj{\eqs}{\prog}{\Delta}$ is valid.
  \end{enumerate}
\end{theorem}

\section{Extended Examples}
\label{section-examples}

In Section \ref{section-smt} we gave the example of 3-party additive
sharing for any field $\fieldp{p}$. In Sections \ref{section-ipj} and
\ref{section-metalangty} we discussed the example of YGC input encoding.
Now we consider examples that illustrate the versatility of the
$\TirName{Encode}$ typing rule for confidentiality, the use of
postconditions for correctness, and the use of both pre-~and
postconditions for integrity.

To evaluate our system especially regarding performance, using the
Satisfiability Modulo Fields theory in cvc5 on an Apple M1 processor,
we have verified correctness of additive sharing protocols in fields
up to size $\fieldp{2^{31} - 1}$ (an approximation of 32-bit integers)
which takes $\sim$7ms to compute for 3 parties and not significantly more
than for, e.g., $\fieldp{2}$. We have also verified the semi-homomorphic
encryption properties leveraged in the BDOZ circuit library functions
(Figures \ref{fig-bdozsum} and \ref{fig-bdozmult}) in fields up to
size $\fieldp{2^{31} - 1}$ which take $\sim$8ms for the property used for
the postcondition of the $\ttt{mult}$ gate. This is categorically
better performance than can be achieved using previous brute-force
methods in $\fieldp{2}$ \cite{skalka-near-ppdp24}.

\subsection{Confidentiality Examples}

\begin{fpfig}[t]{2-party GMW circuit library.}{fig-gmw}
{\footnotesize
  \begin{verbatimtab}
encodegmw(n, i1, i2) {
    m[n]@i2 := (s[n] + r[n])@i1;
    m[n]@i1 := r[n]@i1
}

andtablegmw(x, y, z) {
    let r11 = r[z] + (m[x] + 1) * (m[y] + 1) in
    let r10 = r[z] + (m[x] + 1) * (m[y] + 0) in
    let r01 = r[z] + (m[x] + 0) * (m[y] + 1) in
    let r00 = r[z] + (m[x] + 0) * (m[y] + 0) in
    { row1 = r00; row2 = r01; row3 = r10; row4 = r11 }
}

andgmw(z, x, y) {
   let table = andtablegmw(x,y,z) in
   m[z]@2 := OT4(m[x], m[y], table, 2, 1);
   m[z]@2 as ~((m[x]@1 + m[x]@2) * (m[y]@1 + m[y]@2)) + r[z]@1;
   m[z]@1 := r[z]@1
}
post:
{  m[z]@1 + m[z]@2 == (m[x]@1 + m[x]@2) * (m[y]@1 + m[y]@2) }
\end{verbatimtab}
}
\end{fpfig}

\begin{comment}
xorgmw(z, x, y) {
        m[z]@1 := (m[x] + m[y])@1; m[z]@2 := (m[x] + m[y])@2
}

decodegmw(z) {
        p["1"] := m[z]@1; p["2"] := m[z]@2;
        out@1 := (p["1"] + p["2"])@1;
        out@2 := (p["1"] + p["2"])@2
}
\end{comment}

In the 2-party GMW protocol \cite{evans2018pragmatic}, another garbled
binary circuit protocol, values are encrypted in a manner similar to
BDOZ as described in Section \ref{section-ipj}. In our definition of
GMW we use the convention that shared values $\macgv{\mesg{w}}$ are
identified by strings $w$ and encoded as shares $\mx{w}{1}$ and
$\mx{w}{2}$.  As for YGC, ciphertypes reflect the confidentiality of
GMW input encodings as defined in the $\ttt{encodegmw}$ function in
Figure \ref{fig-gmw}. No programmer annotation is needed given the
syntax of the function body. More interestingly, the $\ttt{andgmw}$
gate definition shows how a programmer hint can express the relevant
confidentiality property of the output share $\mx{z}{2}$ on client 2,
using the $\ttt{as}$ casting. The non-trivial equivalence can be
verified by SMT once during verification, and subsequently
confidentiality is expressed in the dependent type of the function (as
for $\ttt{encode}$).

Note also that the $\ttt{andgmw}$ function is decorated with a
postcondition that expresses the correctness property of the
function. Although not strictly necessary for confidentiality, this
annotation can at least help eliminate bugs, and also can be used
compositionally for whole-program correctness properties. As for
confidentiality hints, this postcondition needs to be verified only
once at the point of definition and subsequently is guaranteed to hold
for any application.

\subsection{Integrity Examples}

\begin{fpfig}[t]{2-party BDOZ circuit library: sum gates and secure opening.}{fig-bdozsum}
{\footnotesize
\begin{verbatimtab}
   pre: { m[x++"m"]@i2 == m[x++"k"]@i1 + (m["delta"]@i1 * m[x++"s"])@i2 }
   _open(x,i1,i2){
     m[x++"exts"]@i1 := m[x++"s"]@i2;
     m[x++"extm"]@i1 := m[x++"m"]@i2;
     assert(m[x++"extm"] = m[x++"k"] + (m["delta"] * m[x++"exts"]))@i1;
     m[x]@i1 := (m[x++"exts"] + m[x++"s"])@i1
   }

  pre: { m[x++"m"]@1 == m[x++"k"]@2 + (m["delta"]@2 * m[x++"s"])@1 /\
         m[x++"m"]@2 == m[x++"k"]@1 + (m["delta"]@1 * m[x++"s"])@2 }
  open(x) { _open(x,1,2); _open(x,2,1) }
   
  _sum(z, x, y, i){
     m[z++"s"]@i := (m[x++"s"] + m[y++"s"])@i;
     m[z++"m"]@i := (m[x++"m"] + m[y++"m"])@i;
     m[z++"k"]@i := (m[x++"k"] + m[y++"k"])@i
  }

  pre: { m[x++"m"]@2 == m[x++"k"]@1 + (m["delta"]@1 * m[x++"s"])@2 /\
         m[y++"m"]@2 == m[y++"k"]@1 + (m["delta"]@1 * m[y++"s"])@2 /\
         m[x++"m"]@1 == m[x++"k"]@2 + (m["delta"]@2 * m[x++"s"])@1 /\
         m[y++"m"]@1 == m[y++"k"]@2 + (m["delta"]@2 * m[y++"s"])@1 }
  sum(z,x,y){ _sum(z,x,y,1); _sum(z,x,y,2) }
  post: { m[z++"m"]@2 == m[z++"k"]@1 + (m["delta"]@1 * m[z++"s"])@2 /\
          m[z++"m"]@1 == m[z++"k"]@2 + (m["delta"]@2 * m[z++"s"])@1 } 
\end{verbatimtab}
}
\end{fpfig}

\begin{fpfig}[t]{2-party BDOZ circuit library: multiplication gates.}{fig-bdozmult}
{\footnotesize
\begin{verbatimtab}
  _mult(z, x, y, i) {
      m[z++"s"]@i :=
        (m[z++"d"] * m[y++"s"] + -m[z++"e"] * m[z++"as"] + m[z++"cs"])@i;
      m[z++"m"]@i :=
        (m[z++"d"] * m[y++"m"] + -m[z++"e"] * m[z++"am"] + m[z++"cm"])@i;
      m[z++"k"]@i :=
        (m[z++"d"] * m[y++"k"] + -m[z++"e"] * m[z++"ak"] + m[z++"ck"])@i
  }
  
  pre: { m[x++"m"]@2 == m[x++"k"]@1 + (m["delta"]@1 * m[x++"s"])@2 /\
         m[y++"m"]@2 == m[y++"k"]@1 + (m["delta"]@1 * m[y++"s"])@2 /\
         m[z++"am"]@2 == m[z++"ak"]@1 + (m["delta"]@1 * m[z++"as"])@2 /\
         m[z++"bm"]@2 == m[z++"bk"]@1 + (m["delta"]@1 * m[z++"bs"])@2 /\
         m[z++"cm"]@2 == m[z++"ck"]@1 + (m["delta"]@1 * m[z++"cs"])@2 /\
         m[z++"cs"]@1 + m[z++"cs"]@2 ==
            (m[z++"as"]@1 + m[z++"as"]@2) * (m[z++"bs"]@1 + m[z++"bs"]@2)}
  mult(z,x,y) {
      sum(z++"d", x, z++"a");
      open(z++"d");
      sum(z++"e", y, z++"b");
      open(z++"e"); 
      _mult(z,x,y,1); _mult(z,x,y,2)
  }
  post: {  m[z++"m"]@2 == m[z++"k"]@1 + (m["delta"]@1 * m[z++"s"])@2 /\
           m[z++"s"]@1 + m[z++"s"]@2 ==
              (m[x++"s"]@1 + m[x++"s"]@2) * (m[y++"s"]@1 + m[y++"s"]@2) } 
  
\end{verbatimtab}
}
\end{fpfig}

% hints for confidentiality
\begin{comment}
      m[z++"ds"]@1 as m[x++"s"]@2 + r[z++"as"]@2;
      m[z++"ds"]@2 as m[x++"s"]@1 + r[z++"as"]@1;
      m[z++"ms"]@2 as m[z++"k"]@1 + (m["delta"]@1 * m[z++"ds"]@2);
      m[z++"ms"]@1 as m[z++"k"]@2 + (m["delta"]@2 * m[z++"ds"]@1);
\end{comment}

\begin{comment}
  pre: { m[x++"m"]@2 == m[x++"k"]@1 + (m["delta"]@1 * m[x++"s"])@2 /\
         m[y++"m"]@2 == m[y++"k"]@1 + (m["delta"]@1 * m[y++"s"])@2 /\
         m[x++"m"]@1 == m[x++"k"]@2 + (m["delta"]@2 * m[x++"s"])@1 /\
         m[y++"m"]@1 == m[y++"k"]@2 + (m["delta"]@2 * m[y++"s"])@1 /\
         m[z++"am"]@2 == m[z++"ak"]@1 + (m["delta"]@1 * m[z++"as"])@2 /\
         m[z++"bm"]@2 == m[z++"bk"]@1 + (m["delta"]@1 * m[z++"bs"])@2 /\
         m[z++"cm"]@2 == m[z++"ck"]@1 + (m["delta"]@1 * m[z++"cs"])@2 /\
         m[z++"am"]@1 == m[z++"ak"]@2 + (m["delta"]@2 * m[z++"as"])@1 /\
         m[z++"bm"]@1 == m[z++"bk"]@2 + (m["delta"]@2 * m[z++"bs"])@1 /\
         m[z++"cm"]@1 == m[z++"ck"]@2 + (m["delta"]@2 * m[z++"cs"])@1 /\
         m[z++"cs"]@1 + m[z++"cs"]@2 ==
         (m[z++"as"]@1 + m[z++"as"]@2) * (m[z++"bs"]@1 + m[z++"bs"]@2)}
\end{comment}

Returning to the example of malicious-secure 2-party BDOZ arithmetic
circuits begun in Section \ref{section-ipj}, in Figure \ref{fig-bdozsum} we
define $\ttt{sum}$ and $\ttt{open}$ functions. The latter implements
``secure opening''-- each party sends it's local share of a
global value $\macgv{\mesg{w}}$ along with its MAC to the other party,
which is authenticated via the $\macbdoz{w}$ check (in \ttt{\_open}),
and then each party reconstructs $\macgv{\mesg{w}}$. The preconditions
of $\ttt{open}$ specify the authentication requirements. 

The $\ttt{sum}$ function implements an addition gate. This is
non-interactive-- each party just sums its local shares
of the two values. The pre- and postcondition annotations
on $\ttt{sum}$ express the additive homomorphism associated
with this encryption scheme-- the sum of MACs of the input
shares is a valid MAC for the output share on each client,
which can be checked using the sum of keys of the input shares.

In BDOZ a pre-processing phase is assumed where initial input secrets
are shared along with their associated MACs and keys. This can be
expressed in $\eqspre$ for input secrets $\sx{\ttt{"x"}}{1}$ and
$\sx{\ttt{"y"}}{2}$, for example, which subsume the following
constraints on shares:
{\footnotesize$$
\begin{array}{l}
\mx{\ttt{"xs"}}{2} \eop \sx{\ttt{"x"}}{1} \fminus \rx{\ttt{"x"}}{1} \wedge 
\mx{\ttt{"xs"}}{1} \eop \rx{\ttt{"x"}}{1} \wedge \\
\mx{\ttt{"ys"}}{1} \eop \sx{\ttt{"y"}}{2} \fminus \rx{\ttt{"y"}}{2} \wedge 
\mx{\ttt{"ys"}}{2} \eop \rx{\ttt{"y"}}{2} 
\end{array}
$$}
and the following constraint on keys and MACs for authentication
of $\sx{\ttt{"x"}}{1}$ (and similarly for $\sx{\ttt{"y"}}{2}$):
{\footnotesize$$
\begin{array}{l}
\mx{\ttt{"delta"}}{1} \eop \rx{\ttt{"delta"}}{1} \wedge
\mx{\ttt{"xk"}}{1} \eop \rx{\ttt{"xk"}}{1} \wedge\\
\mx{\ttt{"xm"}}{2} \eop \mx{\ttt{"xk"}}{1} \fplus (\mx{\ttt{"delta"}}{1} * \mx{\ttt{"xs"}}{2})
\end{array}
$$}
Given these global preconditions, a malicious secure opening of $\sx{\ttt{"x"}}{1} +
\sx{\ttt{"y"}}{2}$ would be obtained as
$\ttt{sum}(\ttt{"z"},\ttt{"x"},\ttt{"y"}); \ttt{open}(\ttt{"z"})$,
which type checks.

A common approach to implementing multiplication gates in a BDOZ
setting is to use \emph{Beaver Triples}. Recall that Beaver triples
are values $a,b,c$ with $a$ and $b$ chosen randomly and $c = a * b$,
unique per multiplication gate, that are secret shared with clients
during pre-processing.  In our encoding we assume the additional
convention that each gate output identifier distinguishes the Beaver
triple, so for example the share of the $a$ value for a gate
$\ttt{"g1"}$ is identified by $\ttt{"g1as"}$, etc.

%, and $\eqspre$
%subsumes the following constraints for the $a,b,c$ values of gate
%$\ttt{"g1"}$.  {\footnotesize$$
%\begin{array}{l}
%\mx{\ttt{"g1as"}}{1} \eop \rx{\ttt{"g1as"}}{1}\ \wedge 
%\mx{\ttt{"g1bs"}}{1} \eop \rx{\ttt{"g1bs"}}{1}\ \wedge\\
%\mx{\ttt{"g1as"}}{2} \eop \rx{\ttt{"g1as"}}{2}\ \wedge 
%\mx{\ttt{"g1bs"}}{2} \eop \rx{\ttt{"g1bs"}}{2}\ \wedge\\
%\mx{\ttt{"g1cs"}}{1} \eop \\
%\qquad ((\mx{\ttt{"g1as"}}{1} \fplus \mx{\ttt{"g1bs"}}{2})\ \ftimes\\
%\qquad\phantom{(}(\mx{\ttt{"g1bs"}}{1} \fplus \mx{\ttt{"g1bs"}}{2})) \fminus \rx{\ttt{"g1cs"}}{2}\ \wedge\\
%\mx{\ttt{"g1cs"}}{2} \eop \rx{\ttt{"g1cs"}}{2}
%\end{array}
%$$}

The definition of the $\ttt{mult}$ gate is presented in Figure
\ref{fig-bdozmult},
%where we present just client 1's side of the
%protocol $\ttt{\_mult1}$ (client 2's side is nearly symmetric but with
%a small variation-- refer to
%\cite{evans2018pragmatic,10.1007/978-3-030-68869-1_3} for more
%details).
%We could call this function on input secrets
%$\sx{\ttt{"x"}}{1}$ and $\sx{\ttt{"y"}}{2}$ in gate $\ttt{"g1"}$ as
%$\ttt{mult}(\ttt{"g1"},\ttt{"x"},\ttt{"y"})$, for example, or embed
%this gate internally in a circuit.
%
As for $\ttt{sum}$, the preconditions of $\ttt{mult}$ express
the expected MAC properties of input shares, as well as the
expected Beaver Triple property, and its postcondition
fexpresses the semi-homomorphic encryption property of
the resulting share, MAC, and key after gate execution-- specifically,
it preserves the BDOZ authentication property. Finally, the
postcondition of $\ttt{mult}$ expresses the
correctness property of the multiplication gate. In any case,
integrity of any circuit constructed from the $\ttt{sum}$ and
$\ttt{mult}$ library functions will require little SMT overhead
due to compositional verification, once their pre- and postconditions
are verified by type checking.

\section{Conclusion and Future Work}

In this paper we have defined confidentiality and integrity type
systems for the $\metaprot/\minifed$ language model. These type
systems resemble previous security type systems in probabilistic
settings, but their accuracy is boosted by the integral use of SMT
verification.  Our work significantly improves automation,
scalability, and flexibility of static analysis in the
$\metaprot/\minifed$ as compared to prior work.

We see several promising avenues for future work. Concurrency will be
important to capture common MPC idioms such as commitment and circuit
optimizations and to consider the UC security model
\cite{evans2018pragmatic,viaduct-UC}. We also hope to improve our type
system to enforce the declassification bounds determined by ideal
functionalities, beyond gradual release.

\bibliographystyle{splncs04}
\bibliography{logic-bibliography,secure-computation-bibliography}

\section*{Appendix}

\subsubsection{Confidentiality Type Soundness}

%\begin{definition}
%  Given $\prog$ and $X_1,X_2 \subseteq \vars{\prog}$, we say that
%  $X_1$ \emph{interferes with} $X_2$
%  iff $\progd(\prog) \not\vdash X_1 * X_2$.
%\end{definition}

The following Lemma enumerates critical properties of dependence of
views on input secrets and random variables, and follows by standard
properties of pmfs.
\begin{lemma}
  \label{lemma-interference}
  Given $\prog$ with $\iov(\prog) = S \cup V \cup O$ and $\flips(\prog) = R$,
  for all $x \in S \cup R$ and $M \subseteq V$ exactly one of the following conditions holds:
  \begin{enumerate}[\hspace{5mm}i.]
  \item $\notsep{\progd(\prog)}{\setit{x}}{M}$
  \item There exists $X \subseteq S \cup R$ such that
    $\notsep{\progd(\prog)}{\setit{x} \cup X}{M}$ while
    $\sep{\progd(\prog)}{\setit{x}}{M}$ and $\sep{\progd(\prog)}{X}{M}$,
    and we say that $x$ encodes $X$ in $M$.
  \item Neither condition (i) nor (ii) hold, implying $\sep{\progd(\prog)}{\setit{x}}{M}$.
  \end{enumerate}
\end{lemma}

%\begin{lemma}
%  \label{lemma-eqsprogsound}
%  If $\toeq{\prog} \models \eqs$ and $\eqs \models \eqs'$, then
%  for all $\store \in \runs(\prog)$ we have $\store \models \eqs'$.
%\end{lemma}

The following Definition characterizes the crucial representation invariant
preserved by confidentiality type checking. 
\begin{definition}
  \label{definition-sound}
  Given $\prog$ with $\iov(\prog) = S \cup M \cup P \cup O$
  and $\flips(\prog) = R$, we say that
  \emph{$\Gamma$ is sound for $\prog$} iff for all $M' \subseteq M$
      and $x \in S \cup R$, the following
      conditions hold:
  \begin{enumerate}[\hspace{5mm}i.]
  \item  if $\notsep{\progd(\prog)}{\setit{x}}{M'}$
    there exists $\ty$ with $\leakj{\Gamma}{M'}{\ty}$ and $x \in \ty$.
    \item  if there exists  $X \subseteq S \cup R$ such that
      $\setit{x}$ encodes $X$ in $M'$, 
      then $\leakj{\Gamma}{M'}{\ty}$ and
      $\leakclose{\Gamma}{\ty \cup X}{\ty'}$ and
      $x \in \ty'$.
  \end{enumerate}
\end{definition}
Now we can prove that this invariant is preserved by type checking.
\begin{lemma}
  \label{lemma-cpjsound}
  If $\cpj{R}{\eqs}{\eqspre \wedge \toeq{\prog}}{\Gamma}$ is valid then $\Gamma$ is
  sound for $\prog$.
\end{lemma}
\begin{proof}
  Unrolling Definition \ref{definition-sound}, proceed by induction on $M'$.
  In case $M' = \varnothing$ the result follows trivially. Otherwise
  $M' = M_0 \cup \{ \mx{w}{\cid} \}$ and wlog $\Gamma = \Gamma_0; \mx{w}{\cid} : \ty_0$
  and $\prog = \prog_0;\xassign{\mx{w}{\cid}}{\be}{\cid_0}$, where $R = R_0;R_1$ and:
  \begin{mathpar}
    \cpj{R_0}{\eqs}{\eqspre \wedge \toeq{\prog_0}}{\Gamma_0}

    \eqj{R_1}{\eqs}{\toeq{\elab{\be}{\cid}}}{\ty_0}
  \end{mathpar}
  are both valid, and by the induction hypothesis $\Gamma_0$ is sound for $\prog_0$.
  We can then proceed by a second induction on the derivation of $\ty_0$, where there
  are two subcases defined by the $\TirName{DepTy}$ and $\TirName{Encode}$ rules.
  In the former case the result is immediate since all variables occurring in
  $\toeq{\elab{\be}{\cid}}$ are elements of $\ty_0$ by definition. In the latter case,
  we have:
  \begin{mathpar}
    \eqs \models \toeq{\elab{\be}{\cid}} \eop \phi' \fminus \rx{w}{\cid}

    \ty_0 = \cty{\rx{w}{\cid}}{\ty_1}

    R_1 = R_2;\setit{\rx{w}{\cid}}

    \eqj{R_2}{\eqs}{R_2}{\ty_1}
  \end{mathpar}
  These facts imply that $\rx{w}{\cid}$ is used as a one-time-pad to encode
  $\vars(\phi)$ in $M'$ 
  \cite{barthe2019probabilistic}. Observing that:
  $$
  \leakclose{\Gamma}{\setit{\cty{\rx{w}{\cid}}{\ty_1}, \rx{w}{\cid}}}{\ty'}
  $$
  The result in this subcase follows by the (second) induction hypothesis. \qed
\end{proof}
On the basis of the above, proving the main result is straightforward.
\begin{proof}[Theorem \ref{theorem-cpj}]
  Immediate by Lemmas \ref{lemma-cpjsound} and \ref{lemma-interference}. \qed
\end{proof}

\subsubsection{Integrity Type Soundness}

A key observation is that, to undermine integrity adversarial strategies
must be different from passive strategies, in particular at least one
adversarial input must be different than any used in a passive strategy
(an aberration). 
\begin{definition}
  Given $H,C$ and $\prog$ with $\iov(\prog) = S \cup V \cup O$,
  we say $x \in V \cup O$ has an \emph{aberration} iff
  $\exists \adversary, \store \in \botruns(\prog)$ such that
  $\store(x) \ne \bot$ and $\neg\exists \store' \in \runs(\prog)$
  with  $\store_{S_H} = \store'_{S_H}$ and
  $\store(x) = \lcod{\store',\be}{\cid}$ where
  $\xassign{x}{\be}{\cid} \in \prog$.
\end{definition}
It is straightforward to prove that if an aberration can
exist, then it will be assigned $\lolab$ integrity in our
analysis. 
\begin{lemma}
  \label{lemma-aberration-low}
  Given protocol $\prog$ with
  $\views(\prog) = V$, if 
  $\ipj{\eqs}{\prog}{\Delta}$ is valid
  for some $\eqs$ and $x \in \cinputs$ has an
  aberration for some $H,C$, then $\cheatj{\Delta}{H,C}{\seclev}$
  and $\seclev(x) = \lolab$.
\end{lemma}
\begin{proof}
  Given $H,C$, \emph{every} adversarial input is assigned $\lolab$
  integrity by definition of the $\TirName{Mesg}$ rule and the
  \emph{leak} judgement.  The exception to this is messages that pass
  the $\TirName{MAC}$ check-- success of this check guarantees that
  the message is \emph{not} an aberration, so it can safely be
  switched to $\hilab$ integrity. \qed
\end{proof}
The following Lemma follows trivially by Definition \ref{def-integrity},
since without an aberration adversarial inputs are the same as some passive
strategy.
\begin{lemma}
  \label{lemma-undermined-logic}
  If $\prog$ with $\iov(\prog) = S \cup V \cup O$ does not have
  integrity for some $H,C$, then there exists $\setit{x} \cup V'
  \subseteq \cinputs$ where $\setit{x}$ has an aberration and
  $\notcondsep{\progd(\prog)}{S_H}{\setit{x} \cup V'}{\houtputs \cup O_H}$.
\end{lemma}

\begin{lemma}
  \label{lemma-undermined}
  Given $\prog$ with $\iov(\prog) = S \cup V \cup O$ and
  $\ipj{\eqs}{\prog}{\Delta}$ valid for some $\eqs$.  For all $x \in
  \cinputs$ with an aberration, if
  $\notcondsep{\progd(\prog)}{S_H}{\setit{x} \cup V'}{\houtputs \cup O_H}$
  then there exists $y \in \houtputs \cup O_H$ such that
  $\cheatj{\Delta}{H,C}{\seclev}$ and $\seclev(y) = \lolab$.
\end{lemma}
\begin{proof}
  By Lemma \ref{lemma-aberration-low} and definition of the
  confidentiality type analysis, since if $\setit{x} \cup V'$
  interferes with $\houtputs \cup O_H$ then $\setit{x}$
  must occur in the definition of some $y \in \houtputs \cup O_H$,
  yielding $\seclev(y) = \lolab$ in the integrity label assignment
  judgement defined in Figure \ref{fig-cheatj}. \qed
\end{proof}
On the basis of the above we can prove our main result. 
\begin{proof}[Theorem \ref{theorem-ipj}]
  By Lemmas \ref{lemma-undermined} and \ref{lemma-undermined-logic}. \qed
\end{proof}

\subsection{Type Checking Soundness}

Soundness of the $\metaprot$ type analysis follows mostly by a
straightforward mapping to the confidentiality and integrity type
systems. The most interesting rules are signature verification
($\TirName{Sig}$) and function application ($\TirName{App}$). 
Here the result hinges on the selection of ``fresh'' variables
that guarantee generalization of constraint entailment.
First we obtain a key result about instantiation and
satisfiability of constraints. 
\begin{lemma}
  \label{lemma-fresh}
  Given $w$ and $\cid$ with $\fresh(w,\cid)$, if $\notg{\eqs}[w/y_1,\cid/y_n]$
  is not satisfiable, then there exists no $w',\cid'$ such that
  $\notg{\eqs}[w'/y_1,\cid'/y_n]$ is satisfiable.
\end{lemma}

\begin{proof}
  Let $\eqs = \notg{\eqs}[w/y_1,\cid/y_n]$ and $\eqs' =
  \notg{\eqs}[w'/y_1,\cid'/y_n]$. If $w',\cid'$ are fresh or chosen
  such that $\eqs'$ is an $\alpha$-renaming of $\eqs$, the result is
  immediate.  Otherwise, it must be the case that at least two
  distinct variables $x_1$ and $x_2$ in $\eqs$ are renamed as a single
  variable $x$ in $\eqs'$, given freshness of $w,\cid$. Let
  $\mathit{rename}$ be the function that maps variables in $\eqs$ to
  variables in $\eqs'$ such that $\mathit{rename}(\eqs) = \eqs'$. Let
  $\store'$ be a model of $\eqs'$, and define $\store$ such that:
  $$
  \forall x \in \dom(\vars(E)) . \store(x) \defeq \store'(\mathit{rename}(x)) 
  $$
  Then $\store$ must be a model of $\eqs$, which is a contradiction. \qed
\end{proof}
Now we can extend the preceding Lemma to constraint entailment,
incorporating evaluation of expressions within constraints. 
\begin{lemma}
  \label{lemma-eqs-notg}
  For all $\fresh(\mv_1,\ldots,\mv_n)$, $\notg{E_1}$, $\notg{E_2}$, assuming: 
  \begin{mathpar}
    \notg{\eqs_1}[\mv_1/y_1]\cdots[\mv_n/y_n]  \redx \eqs_1

    \notg{\eqs_2}[\mv_1/y_1]\cdots[\mv_n/y_n]  \redx \eqs_2

    E_1 \models E_2
  \end{mathpar}
  then there exists no $\mv_1',\ldots,\mv_n'$ such that:
    \begin{mathpar}
    \notg{\eqs_1}[\mv_1'/y_1]\cdots[\mv_n'/y_n]  \redx \eqs_1

    \notg{\eqs_2}[\mv_1'/y_1]\cdots[\mv_n'/y_n]  \redx \eqs_2

    E_1' \not\models E_2'
  \end{mathpar}  
\end{lemma}

\begin{proof}
  Given the assumptions, suppose that $E_1' \not\models E_2'$
  is an $\alpha$-renaming of $E_1 \models E_2$. In this
  case the consequence follows immediately. Otherwise,
  suppose on the contrary we can witness
  $E_1' \not\models E_2'$. By definition this means that
  $E_1' \wedge \neg E_2'$ is satisfiable
  but $E_1 \wedge \neg E_2$ is not, which is a contradiction
  by Lemma \ref{lemma-fresh}. \qed 
\end{proof}
Now we can consider the main result.
\begin{proof}[Theorem \ref{theorem-mtj}]
  (Sketch) The result follows by induction on derivation of the
  main type judgement for $\instr$. The \TirName{Mesg} case follows
  by Lemma \ref{lemma-atj-sound}. The \TirName{Encode} and \TirName{App}
  cases follow by definition of the \TirName{Sig} rule and Lemma
  \ref{lemma-eqs-notg}. 
\end{proof}

\end{document}